\documentclass[reqno]{article}
\usepackage[a4paper, margin=1in]{geometry}
\usepackage{graphicx}
\usepackage{dcolumn} 
\usepackage{relsize}
\usepackage{authblk}
\usepackage{anyfontsize}
\usepackage{xfrac}
\usepackage{bm}
\usepackage[utf8]{inputenc}
\usepackage{amsmath}
\usepackage{amssymb}
\usepackage{amsthm}
\usepackage{mathtools}
\usepackage{hyperref}
\hypersetup{colorlinks,allcolors=blue}
\usepackage{physics}
\usepackage{mymacros}
\usepackage{url}
\usepackage{amsfonts}
\usepackage{tikz,pgf, tikz-3dplot}
\usepackage{tabularx,stackengine}
\usepackage{soul}
\newtheorem{conjecture}{Conjecture}
\newtheorem{theorem}{Theorem}[section]
\newtheorem{lemma}[theorem]{Lemma}
\newtheorem{proposition}[theorem]{Proposition}
\newtheorem{definition}[theorem]{Definition}

\usepackage[backend=biber, style=phys]{biblatex}
\DeclareFieldFormat{labelnumberwidth}{\mkbibbrackets{#1} }

\usetikzlibrary{decorations.markings,decorations.shapes,decorations.pathmorphing,svg.path, 3d}
\usetikzlibrary{chains}
\renewcommand{\selectlanguage}[1]{}

\usepackage{tikz}
\usetikzlibrary{decorations.pathmorphing, positioning, shapes.geometric}

\usepackage{tikz}
\usetikzlibrary{decorations.pathmorphing, positioning}

\definecolor{gemLightGray}{HTML}{DCDCDC} 
\definecolor{deepCharcoal}{HTML}{2D2D2D} 
\definecolor{c0}{HTML}{D35400}  
\definecolor{c1}{HTML}{1B5E20}  
\definecolor{c2}{HTML}{1A237E}   
\definecolor{c3}{HTML}{7B241C}      
\definecolor{c4}{HTML}{00695C}     

\tikzset{
  manifold gem/.style={
    v1/.style={
        circle, 
        fill=gemLightGray, 
        draw=deepCharcoal, 
        thick, 
        minimum size=2mm,
        inner sep=0pt
    },
    v2/.style={
        circle, 
        fill=deepCharcoal, 
        draw=deepCharcoal, 
        thick, 
        minimum size=2mm, 
        inner sep=0pt,
        text=white 
    },
    e0/.style={c0, thick},
    e1/.style={c1, thick, dashed},
    e2/.style={c2, thick, dotted},
    e3/.style={c3, thin, double, double distance=.3pt},
    e4/.style={
        c4, 
        semithick,
        decorate, 
        decoration={snake, amplitude=0.2mm, segment length=2.5mm}
    },
    gemEdge/.style={line cap=round, line join=round}
  }
}

\newcommand{\fei}{{\rm F8}_}
\newcommand{\str}{{\rm star}}
\newcommand{\Z}{\mathbb{Z}}
\newcommand{\R}{\mathbb{R}}

\newcommand{\GG}{\mathcal{G}}

\definecolor{darkred}{HTML}{D40000}
\definecolor{lightblue}{HTML}{0066FF}
\definecolor{darkblue}{HTML}{003399}

\begin{document}
\title{Universal entanglement probes of topological order and locally-achiral manifolds}

\author{Yarden Sheffer}
\affil{Department of Condensed Matter Physics, Weizmann Institute of Science, Rehovot 7610001, Israel}

\date{\today}
\maketitle
\begin{abstract}
  We consider the problem of identifying a topological order based on bulk entanglement of the ground-state wavefunction. Previous work showed that some universal information can be extracted from multi-entropy measures, a class of multipartite entanglement measures obtained by applying permutation operators exchanging the degrees of freedom between different replicas of the wavefunction. It remains an open question to what extent such entanglement measures can be used to extract any universal information from the ground state. Here we show that, assuming that the state is well described by a topological field theory, the topological partition function $Z(M)$ of a manifold $M$ can be extracted provided that $M$ satisfies a topological condition which we term ``local achirality". We show that locally-achiral manifolds can be used to extract universal properties of 2+1d topological phases that go beyond the $S$ and $T$ matrices. As a first step towards classifying locally-achiral manifolds, we show that, in four dimensions, such manifolds have vanishing Pontryagin number. We relate this property to the existence of beyond-cohomology time-reversal symmetry protected topological order (T-SPT) in four dimensions.
\end{abstract}
\tableofcontents
\section{Introduction}
Topologically ordered (TO) states are gapped quantum phases of matter which cannot be identified with a trivial gapped phase. In contrast to phases of matter arising in Landau's paradigm, such phases can arise without breaking any (global) symmetry. In the absence of such symmetry breaking, TO ground states are distinguished by different patterns of entanglement. It has long been known that the entanglement entropy can be used as a probe of topological order in two dimensions, via an order one correction to the area law termed the topological entanglement entropy (TEE) \cite{kitaev_2006_topological,Levin_Wen_2006}. Later work extended the TEE to higher dimensions \cite{grover2011entanglement}. It is also known, however, that the TEE cannot capture the full complexity of possible TO states, and that distinct phases can have the same TEE. As a concrete example, the toric-code and the Ising phases of the Kitaev Honeycomb model \cite{kitaev2006anyons} have the same TEE but different anyon statistics. Recent years have seen a surge in proposed multipartite entanglement measures that can probe properties of a TO phase beyond the TEE. Such measures could detect, for example, the presence of gapless edge modes \cite{siva_universal_2022, liu_multi_2024}, the chiral central charge \cite{kim_chiral_2022,kim2022modular, sheffer2025probing, gass2025r}, the anyon quantum dimensions \cite{Liu_2024}, and the topological spins \cite{sheffer2025extracting}. 

Many of the recent results were obtained via topological multi-entropy measures. Originally considered in the context of the AdS/CFT correspondence \cite{Gadde_Krishna_Sharma_2022, penington_fun_2023}, such measures are multipartite generalizations of the R\'enyi entropy. Multi-entropy measures are defined by considering multiple copies (replicas) of a wavefunction $\ket\y$, and the expectation value of the operator applying permutations of the local degrees of freedom between the replicas, applying different permutations on different subsystems (when $\ket\y$ is defined on a finite-dimensional lattice, the subsystems are different spatial regions). The power of multi-entropy measures is that they allow for different universal invariants to be obtained by varying the choice of permutation operators. Moreover, in the context of topological phases, the resulting universal quantities admit an interpretation in terms of the topological partition function $Z(M)$ evaluated on a certain manifold $M$ (which depends on the choice of permutations, as we explain below) \cite{dong_topological_2008, sheffer2025extracting, del2026multipartite}. 

The aim of this work is to consider a basic question about multipartite entanglement measures of TO: can such measures be used to completely determine a topological state? The construction of \cite{del2026multipartite} shows that, for any manifold $M$, the magnitude $\abs{Z(M)}$ can be extracted universally via multi-entropy measures. It is less clear, however, that the argument $\arg(Z(M))$ can be obtained universally as well. This is required for a full characterization of the phase, as, for example, a phase and its time-reversed partner will have the same magnitude for any partition function. Similarly, in the context of invertible phases (such as symmetry-protected topological orders (SPTs)), the partition function is of magnitude 1, and the entire information is carried by the argument. The central contribution of our work here is to define a class of manifolds, which we term ``locally-achiral" manifolds, for which $\arg(Z(M))$ can be extracted using multi-entropy measures. This condition follows from the idea that non-universal terms in the multi entropy come from local contributions. The local achirality condition ensures that these local contributions cannot contribute to the phase. While it remains unclear which manifolds are locally-achiral (it is possible, for example, that all 3-manifolds are), we present a construction that results in a large number of locally-achiral 3-manifolds. This shows that locally-achiral manifolds are quite ubiquitous.

In the problem of distinguishing between topological states, a particular case of interest is that of the Mignard-Schauenburg (MS) examples \cite{Mignard_Schauenburg_2021}. These are 3-dimensional topological theories which are distinct yet cannot be distinguished using their modular invariants (the $S$ and $T$ matrices). To our knowledge, these are the only known examples which cannot be distinguished using the set of entanglement measures presented in \cite{sheffer2025extracting}. We present locally-achiral manifolds whose partition function goes beyond the modular invariants, and show that their partition function can distinguish infinitely many (though not all) of the MS examples. We therefore conjecture that any two topological theories in 2+1D can be distinguished by the partition function on locally-achiral manifolds, so that there always exists a multi-entropy measure that distinguishes between them. Beyond this example, it is interesting to ask whether the notion of local achirality applies to any manifold. In higher dimensions, when the dimension is $0\mod 4$, we argue that the Pontryagin numbers are obstructions to a manifold being locally-achiral. This is related to the existence of time-reversal symmetry protected topological orders whose partition function is given by the Pontryagin number mod 2.

\section{Preliminaries}
\subsection{From multi-entropy measures to topological manifolds}
\label{sec: multi-ent}
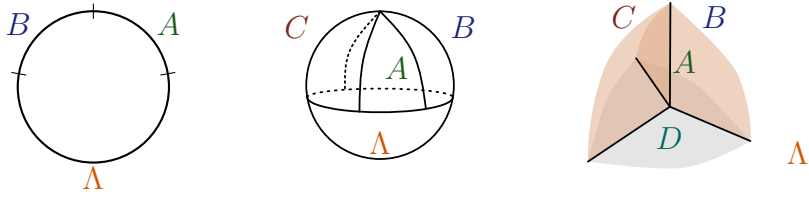
\begin{figure}
  \begin{center}
  \begin{tikzpicture}
    \node at (0,0){
      \begin{tikzpicture}[manifold gem]
          \def\r{1}
          \def\d{.1}
          \draw[thick] (1,0) arc (0:360:\r);
          \draw (10:\r-\d) --++ (10:2*\d);
          \draw (90:\r-\d) --++ (90:2*\d);
          \draw (180-10:\r-\d) --++ (180-10:2*\d);
          \node[c1] at (40:1.3) {\large $A$};
          \node[c2] at (140:1.3) {\large $B$};
          \node[c0] at (-90:1.2) {\large $\Lambda$};
        \end{tikzpicture} };
    \node at (3.8,.2){
      \begin{tikzpicture}[scale=1.1,manifold gem]
          \node at (0,0) {\includegraphics[scale=.8]{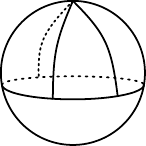}};
          \node[c0] at (0,-.7) {\large $\Lambda$};
          \node[c1] at (0.2,.2) {\large $A$};
          \node[c2] at (1,.7) {\large $B$};
          \node[c3] at (-1,.7) {\large $C$};
        \end{tikzpicture} };
    \node at (8,.2){
      \begin{tikzpicture}[manifold gem]
          \node at (0,0) {\includegraphics[scale=.8]{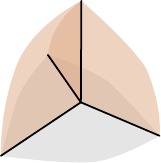}};
          \node[c0] at (1.7,-.9) {\large $\Lambda$};
          \node[c1] at (0.2,.3) {\large $A$};
          \node[c2] at (.6,.9) {\large $B$};
          \node[c3] at (-.6,.9) {\large $C$};
          \node[c4] at (0,-.7) {\large $D$};
        \end{tikzpicture} };
  \end{tikzpicture}
  \end{center}
  \caption{Partition of the 1,2 and 3-sphere into regions. The 3-sphere is considered as $\mathbb{R}^3$ compactified at infinity.}\label{fig: regions}
\end{figure}
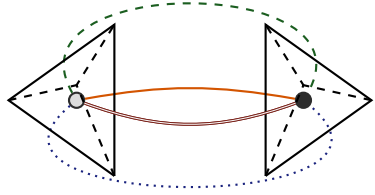
\begin{figure}
  \begin{center}
    \begin{tikzpicture}
      \begin{scope}[manifold gem]
        \node[v1] (1) at (-1.5,1.0) {};
        \node[v2] (2) at (1.5,1.0) {};
        \draw[e0] (1) to[bend left=10] (2);
        \draw[e1] (1) to[bend left=120,looseness=1.5] (2);
        \draw[e2] (1) to[bend right=140,looseness=1.8] (2);
      \end{scope}
      \foreach \x in {-1,1}{
        \begin{scope}[shift={(\x,0)}, xscale=\x,every path/.append style={thick}]
        \coordinate (a) at (0,0);
        \coordinate (b) at (1.4,1);
        \coordinate (c) at (0,2);
        \coordinate (d) at (.5,1.2);
        
        \draw (a) -- (b) -- (c) -- (a);
        \draw[dashed] (a) -- (d) -- (c);
        \draw[dashed] (b) -- (d);
      \end{scope}
      \begin{scope}[manifold gem]
      \draw[e3] (1) to[bend right=20] (2);
    \end{scope}
    }
    \end{tikzpicture}
  \end{center}
  \caption{The ``trivial" multi entropy $\ev{\y|\y}$ defines the manifold $S^D$ by gluing two $D$-simplexes, and the appropriate gem, here shown for $D=3$.}\label{fig: multi-ent-example}
\end{figure}

We begin by setting up the problem. We consider a wavefunction $\ket\psi$ defined on a lattice approximating the $d$-dimensional sphere and assume that $\ket\psi$ is a gapped ground state of a local Hamiltonian. We are interested in characterizing the topological phase which $\psi$ belongs to. One framework of characterizing the universal properties of the phase is via multi-entropy measures, defined as follows: divide the sphere into $d+2$ regions $A_i$, such that any subset of $d+1$ regions intersects at a vertex (see Fig. \ref{fig: regions}). Consider $R$ replicas of the wavefunction and $d+2$ permutations of $R$ elements $\pi_{A_1},...,\pi_{A_{d+2}}$. The multi-entropy is then defined as
\begin{equation}
  \mathcal{M}\qty(\y)=\mel{\psi^{\otimes R}}{\pi_{A_1}\cdots\pi_{A_{d+2}}}{\y^{\otimes R}},
  \label{eq: multi-ent-def}
\end{equation}
where now the operator $\pi_{A_i}$ permutes the local degrees of freedom within region $A_i$ according to the permutation element \footnote{The choice of $\ket\y$ defined on the sphere is made for ease of presentation. Due to locality, the same results will apply with any other topology, with $A_1,...,A_{d+1}$ forming a disk-shaped region, and $A_{d+2}$ being the rest of space.}. Any multi-entropy is invariant under local unitaries acting inside the regions $A_i$, and as such defines a measure of entanglement between the regions. If $\ket{\psi}$ has a description in terms of a TQFT, the multi-entropy is related to the partition function of the TQFT evaluated on a certain manifold $M$ \cite{dong_topological_2008}. This manifold is obtained from the permutation elements using the following procedure: take $R$ ``ket" and $R$ ``bra" simplexes of dimension $D=d+1$. Label the faces of the bra and ket simplexes by the regions $A_i$, with opposite orientation between the bra and ket. The manifold $M$ is then obtained by gluing each $A_i$ face of the ket simplex $r$ with the $A_i$ face of the bra simplex $\pi_{A_i}(r)$ (see an example in Fig. \ref{fig: multi-ent-example}). For ``generic" wavefunctions, it is conjectured that the following relation holds for the multi-entropy
\begin{equation}
  \mathcal{M}(\y)=Z(M,\psi)\times\text{local terms}
  \label{eq: multi-ent-tqft}
\end{equation}
where $Z(M,\psi)$ is the partition function of the TQFT associated with $\y$ on the $D$-manifold $M$ \footnote{Some care has to be taken for states with nonzero thermal Hall conductance, see \cite{sheffer2025probing}. We will ignore this case throughout this work.}. The local terms are obtained from contributions from $\psi$ that come from ($d+1$)-partite entanglement, or, geometrically, from lower-dimensional contributions. To state this conjecture slightly more precisely, we consider a ``multi-entropy probe" $\mathcal{F}(\y)$ to be a function of different multi-entropy measures defined on the $d+2$ regions, such that $\mathcal{F}$ is multiplicative under stacking. An example of such a probe, which will be used extensively throughout this work, is the phase contribution
\begin{equation}
  \mathcal{P}(\y) = \frac{\mathcal{M}(\y)}{\abs{\MM(\y)}}.
  \label{eq: phase-probe}
\end{equation}

\begin{conjecture}[Validity of multi-entropy probes]
  Assume that $\mathcal{F}$ is a multi-entropy probe in $d$ dimensions, such that $\mathcal{F}(\y)=1$ for states supported on any $d+1$ of the regions. Then, for regions of linear size $L$, a generic (quasi-)local quantum circuit $U$ of depth $t$, and generic gapped ground state of a local Hamiltonian $\y$, we have $\mathcal{F}(\y)\approx\mathcal{F}(U\y)$ with errors small as $L/t\to \infty$. Furthermore, $\mathcal{F}$ can be calculated by replacing any measure $\mathcal{M}$ in this expression with the appropriate TQFT partition function $Z(M)$ \footnote{One can also consider a symmetric version of this conjecture, useful in the context of SPTs: we would then require $\ket\psi$ to be symmetric under some symmetry $V$ (unitary or antiunitary), and $U$ be a product of $V$-symmetric gates.}.
\end{conjecture}

The condition of triviality on any $d+1$ regions is required, otherwise the conclusion can be trivially violated by stacking with a state supported only around the corner between those regions. Such probes should be thought of as probes that systematically cancel the local contributions in \eqref{eq: multi-ent-tqft}. This conjecture is supported by field-theory arguments \cite{dong_topological_2008, sheffer2025probing} and lattice calculations \cite{del2026multipartite}. If the ``generic" condition is ignored, the claim is known to be incorrect due to spurious contributions \cite{williamson2016spurious, Zou_Haah_2016, Gass_Levin_2024} in certain fine-tuned models. Here, generic should be interpreted as the space of violating $\ket\y$ and $U$ having small measure as $L$ becomes large (in the space of gapped local Hamiltonians and low-depth circuits, respectively). Here we will not attempt to argue for this conjecture, but will rather take it as an axiom. The contributions of the local terms will, however, impose some nontrivial constraints which will be dealt with extensively below. Taking $\mathcal{P}$ as an example, we expect that if $\arg(\MM)$ is trivial on any $d+1$ partite state, then $\arg(\MM)=\arg(Z(M))$.

As was pointed out in \cite{del2026multipartite}, the relation between the permutation operators and the manifold $M$ is related to the notion of a ``gem" of a manifold. We give a few definitions:
\begin{definition}[$n$-graph]
  An \textup{$n$-graph} is a bipartite graph $\GG$ where each vertex is attached to $n$ edges, and the edges colored by $n$ different colors such that the edges attached to each vertex have distinct colors. Using bipartiteness, the vertices acquire a natural 2-coloring.
\end{definition}
The construction above shows that any multi-entropy defines a ($D+1$)-graph with $2R$ vertices. While any such graph defines a $D$-dimensional triangulation, this triangulation will not necessarily define a topological manifold. To obtain a manifold, we need to impose the additional requirement that the local region of each vertex looks like $\mathbb{R}^D$. Define the $i$-residue $\GG_i$ to be the graph obtained from $\GG$ by removing the edges of color $i$. It is a $D$-graph (not necessarily connected) and therefore defines a ($D-1$)-dimensional triangulation. It turns out that $M$ will be a manifold if and only if each connected component of $\GG_i$ is a triangulation of $S^{D-1}$. We therefore define:
\begin{definition}[gem]
  \label{def: gem}
  A \textup{``graph-encoded $n$-manifold"} or an \textup{$n$-gem} is an ($n+1$)-graph $\GG$ such that each connected component of each residue $\GG_i$ is a triangulation of $S^{n-1}$.
\end{definition}
Any gem will therefore define a manifold. Any triangulated manifold (and, therefore, any smooth manifold) can be defined in terms of a gem. We give a few examples of $3$-gems in Fig. \ref{fig: gems-examples}.
\begin{figure}
  \begin{center}
    \begin{tikzpicture}
      \node at (0,0) {
    \begin{tikzpicture}[manifold gem]
      \node[v1] (a) at (0,0) {};
      \node[v2] (b) at (0,2) {};

      \draw[e0] (a) to[bend left=45] (b);
      \draw[e1] (a) to[bend left=15] (b);
      \draw[e2] (a) to[bend left=-15] (b);
      \draw[e3] (a) to[bend left=-45] (b);
  \end{tikzpicture} };
  \node at (4,0) {
    \begin{tikzpicture}[scale=.8,manifold gem]
      \def\r{.6}
      \def\rr{1.4}
      \node[v1] (11) at (\r,\r) {};
      \node[v2] (21) at (\rr,\rr) {};
      \node[v2] (12) at (-\r,\r) {};
      \node[v1] (22) at (-\rr,\rr) {};
      \node[v1] (13) at (-\r,-\r) {};
      \node[v2] (23) at (-\rr,-\rr) {};
      \node[v2] (14) at (\r,-\r) {};
      \node[v1] (24) at (\rr,-\rr) {};

      \draw[e0] (11) to (12);
      \draw[e0] (13) to (14);
      \draw[e0] (21) to (22);
      \draw[e0] (23) to (24);
      \draw[e1] (11) to (21);
      \draw[e1] (12) to (22);
      \draw[e1] (13) to (23);
      \draw[e1] (14) to (24);
      \draw[e2] (11) to (14);
      \draw[e2] (13) to (12);
      \draw[e2] (21) to (24);
      \draw[e2] (23) to (22);
      \draw[e3] (11) to[out=155,in=75, looseness=1.8] (23);
      \draw[e3] (12) to[out=155+90,in=75+90, looseness=1.8] (24);
      \draw[e3] (13) to[out=155+2*90,in=75+2*90, looseness=1.8] (21);
      \draw[e3] (14) to[out=155+3*90,in=75+3*90, looseness=1.8] (22);
  \end{tikzpicture}
};
  \node at (8,0) {
    \begin{tikzpicture}[scale=.8,manifold gem]
      \def\r{.7}
      \def\rr{1.7}
      \foreach \q in {1,3,...,6}
        \node[v1] (1-\q) at (60*\q:\r) {};
      \foreach \q in {2,4,...,6}
        \node[v2] (1-\q) at (60*\q:\r) {};
      \foreach \q in {1,3,...,6}
        \node[v2] (2-\q) at (60*\q:\rr) {};
      \foreach \q in {2,4,...,6}
        \node[v1] (2-\q) at (60*\q:\rr) {};
      \foreach \i in {1,2}{
        \draw[e2] (\i-1) -- (\i-2);
        \draw[e2] (\i-3) -- (\i-4);
        \draw[e2] (\i-5) -- (\i-6);
        \draw[e0] (\i-1) -- (\i-6);
        \draw[e0] (\i-3) -- (\i-2);
        \draw[e0] (\i-5) -- (\i-4);
      }
      \draw[e3] (1-1) to[bend right=40] (2-3);
      \draw[e3] (1-2) to[bend right=40] (2-4);
      \draw[e3] (1-3) to[bend right=40] (2-5);
      \draw[e3] (1-4) to[bend right=40] (2-6);
      \draw[e3] (1-5) to[bend right=40] (2-1);
      \draw[e3] (1-6) to[bend right=40] (2-2);
      \foreach \q in {1,...,6}
        \draw[e1] (1-\q) to (2-\q);

  \end{tikzpicture}
};
\node at (0,-1.6) {(a)};
\node at (4,-1.6) {(b)};
\node at (8,-1.6) {(c)};
  \end{tikzpicture}
  \caption{Example gems for (a) the sphere $S^3$, (b) the projective space $\mathbb{R}P^3$}\label{fig: gems-examples} and (c) the lens space $L(3,1)$.
  \end{center}
\end{figure}
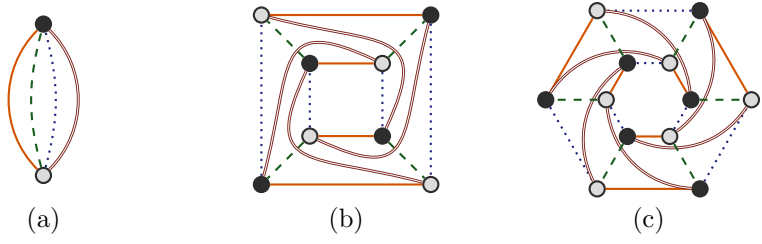

\subsection{Cancellation of local contributions and locally-achiral gems}
We are interested in finding a systematic way for canceling the local terms in \eqref{eq: multi-ent-tqft}, and extracting the universal value $Z(M)$. Generally, from any multi-entropy measure in $d=2$ one can write down a probe that extracts the magnitude $\abs{Z(M)}$ using the cancellation scheme presented in \cite{del2026multipartite}. This is a generalization of the more familiar scheme of extracting the universal term in the entanglement entropy due to Kitaev and Preskill \cite{kitaev_2006_topological}. In many cases of interest, however, characterizing the phase of $\ket\y$ requires extracting the argument $\arg(Z(M))$. This is the case for invertible phases, where $\abs{Z(M)}=1$, and when distinguishing a phase from its time-reversed partner. Canceling the local terms in the phase turns out to be a subtler issue, and will be the main point of interest of this work. 

To begin, it is useful to give a condition for when a gem gives an entanglement measure with trivial phase $\arg(\mathcal{M})=0$.
\begin{definition}
  An $n$-graph $\GG$ will be called \textup{reflection-positive} (RP) if there exists a cut of $\GG$ into two parts, and an involutive automorphism of $\GG$ that exchanges the two parts, fixes each edge on the cut, preserves the edge colors, and exchanges the vertex colors.
\end{definition}
Fig. \ref{fig: rp-gem} presents an example of an RP gem. The usefulness of reflection-positivity lies in the following proposition
\begin{proposition}
  If a multi-entropy measure $\mathcal{M}$ on an $n$-partite state $\psi$ defines a reflection-positive $n$-graph, then $\mathcal{M}(\y)\ge0$ for any state $\y$. 
\end{proposition}
\begin{proof}
  Representing $\ket\psi$ as a tensor with $n$ legs, and consider the tensor network representing $\mathcal{M}(\y)$. The tensors $T_p$ on each side of the partition can be written as $(T_p)_{i_1\cdots i_k}$ where $k$ is the number of legs crossing the partition. Since reflection exchanges vertex colors, it exchanges the tensors $\y$ and $\y^*$. As a result, we have $T_1=T_2^*$. We therefore get 
  \begin{equation}
    \mathcal{M}(\y)=\sum_{i_1,...,i_k}(T_1)_{i_1,...,i_k}(T_1)_{i_1,...,i_k}^* \ge 0.
  \end{equation}
\end{proof}

Note that the gem of $\mathbb{R}P^3$ in Fig. \ref{fig: gems-examples} is not RP. It does, however, satisfy a weaker condition that there exists an automorphism that exchanges vertex colors but preserves edge colors. As a result, the related multi-entropy measure satisfies $\MM_{\mathbb{R}P^3}\in\R$. Also note that any such automorphism defines, topologically, an orientation-reversing involution on the manifold $M$ \footnote{This follows since orientation is defined on each simplex as an ordering of the vertices of the simplex up to an even permutation. Since the orientations should match between two adjacent simplexes, the ordering should be opposite between the two vertex colors, and any symmetry that exchanges them should be orientation-reversing.}. It therefore will not exist whenever the manifold does not admit any such automorphism, that is, when it is a chiral manifold.

To ensure the cancellation of local contributions to the phase, we require that the value of $\arg(\MM)$ is trivial for any state that is supported only on $d+1$ of the regions $A_1,...,A_{d+2}$. The tensor network representing the evaluation of $\MM$ on such a state is the tensor network obtained from $\GG$ with the leg corresponding to region $i$ missing. It is therefore given by $\GG_i$. Therefore, the phase will be trivial for all such locally-entangled states if every $\GG_i$ is RP. This leads us to the main definition of this work:
\begin{definition}[locally-achiral]
  A gem $\GG$ will be called \textup{locally achiral} if the residues $\GG_i$ are RP for any color $i$. A manifold will be called locally-achiral if it admits a locally-achiral gem.
\end{definition}
Note that the residues $\GG_i$ encode the structure of the triangulation around the vertices. A locally-achiral gem is one that has neighborhoods of the vertices that look symmetric under a local orientation-reversing symmetry, even if it might not have such an orientation-reversing symmetry globally. 

Assuming the conjecture above, whenever a manifold $M$ is locally-achiral, we can extract $Z(M,\y)$ using multi-entropy measures. Previous work \cite{sheffer2025extracting} relied on the fact that the lens spaces $L(p,q)$ are chiral but locally achiral (see Fig. \ref{fig: gems-examples}c for an example) to show that the lens space partition functions can be extracted from multi-entropies of the ground state wavefunction. We will later show that not all manifolds are locally achiral, but it remains a question of which manifolds have this property. While we are far from obtaining a general answer to this question, we will make two contributions towards it; both have some interesting implications by themselves. 

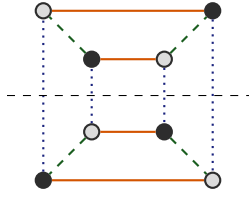
\begin{figure}
  \begin{center}
    \begin{tikzpicture}[scale=.8,manifold gem]
      \def\r{.6}
      \def\rr{1.4}
      \node[v1] (11) at (\r,\r) {};
      \node[v2] (21) at (\rr,\rr) {};
      \node[v2] (12) at (-\r,\r) {};
      \node[v1] (22) at (-\rr,\rr) {};
      \node[v1] (13) at (-\r,-\r) {};
      \node[v2] (23) at (-\rr,-\rr) {};
      \node[v2] (14) at (\r,-\r) {};
      \node[v1] (24) at (\rr,-\rr) {};

      \draw[e0] (11) to (12);
      \draw[e0] (13) to (14);
      \draw[e0] (21) to (22);
      \draw[e0] (23) to (24);
      \draw[e1] (11) to (21);
      \draw[e1] (12) to (22);
      \draw[e1] (13) to (23);
      \draw[e1] (14) to (24);
      \draw[e2] (11) to (14);
      \draw[e2] (13) to (12);
      \draw[e2] (21) to (24);
      \draw[e2] (23) to (22);
      \draw[thin,dashed] (-2,0) -- (2,0);
  \end{tikzpicture}
  \end{center}
  \caption{A reflection-positive 2-gem. The dashed line represents the cut. For any tripartite state $\ket\y$ the resulting tensor network defined by this gem is evaluated as an inner product of the tensors above and below the dashed line, giving a positive result.}
  \label{fig: rp-gem}
\end{figure}

In Sec. \ref{sec: 3-manifolds} we present a scheme for obtaining locally-achiral gems of certain 3-manifolds. In particular, we obtain locally-achiral gems for the manifolds obtained by integer surgery of the figure-8 knot. The partition functions on these manifolds distinguish between some of the theories given by Mignard and Schauenburg \cite{Mignard_Schauenburg_2021} (MS) as examples of anyon theories that cannot be distinguished by modular data. We show that infinitely many of the MS examples can be distinguished in this manner. This result shows that multi-entropy is capable of distinguishing between phases beyond modular data.  

In Sec. \ref{sec: 4-manifolds} we obtain a result in the other direction: we prove that 4-manifolds with nonzero Pontryagin number $p_1$ are not locally-achiral. We conjecture that this is true for any manifold of dimension 0 mod 4 with nonzero Pontryagin number. We justify this conjecture on physical grounds, based on the relation between Pontryagin numbers and time-reversal symmetry protected topological orders (T-SPTs) in these dimensions. 

In Sec. \ref{sec: conclusions} we conclude and discuss a few open questions for future work.

\section{Three-dimensional probes: detection of beyond-modular invariants}
\label{sec: 3-manifolds}

\subsection{The Mignard-Schauenburg examples}
The main goal of this section is to present locally-achiral manifolds whose partition functions detect TQFT data beyond modular data. As a motivating example, we begin by presenting the Mignard-Schauenburg (MS) examples of modular categories that are not determined by modular data. These are constructed as follows: for primes $p,q$ satisfying $p|q-1$, there is a unique nonabelian group of order $pq$ given by 
\begin{equation}
  G_{p,q}=\ev{ a,b|b^p=a^q=1,bab^{-1}=a^s}
  \label{eq:G_pq}
\end{equation}
where $s$ is an integer satisfying $s\neq1, s^p=1\mod q$ (the resulting group does not depend on the choice of $s$, and we will always take an implicit choice of $s$). The third group cohomology is given by $H^3(G_{p,q},U(1))=\Z_p$ and the group cocycles are obtained by composing the quotient map $G_{p,q}\to \Z_p$ with the group cocycles of $\Z_p$. We can therefore label them by an integer $u=0,...,p-1$. As was shown in \cite{Mignard_Schauenburg_2021}, the twisted quantum double categories $D(G_{p,q},u)$ are non-isomorphic for different values of $u$. On the other hand, if $u_1,u_2$ are both 0, quadratic residues, or quadratic nonresidues mod $p$, then there exists a map $\rho$ between the anyon labels of $D(G_{p,q},u_1)$ and $D(G_{p,q},u_2)$, such that $(S_1)_{kl}=(S_2)_{\rho(k)\rho(l)}$, $(T_1)_{kl}=(T_2)_{\rho(k)\rho(l)}$ where $S_i,T_i$ are the $S,T$ matrices corresponding to $u_i$. In other words, knowledge of the $S$ and $T$ matrices does not suffice to distinguish between the theories. 

To begin, we focus on the simplest MS example, with $p=5,q=11$. In this case, the theories with $u=1,4$ have the same modular data, and similarly those with $u=2,3$. In \cite{Delaney_Tran_2018} it was shown that the modular data, together with the values $K_i$ of the anyon diagrams of the figure-8 knot (see Fig. \ref{fig: fig-8}) suffice to distinguish between the $G_{5,11}$ quantum doubles. In fact, a slightly stronger result is true:
\begin{proposition}
  Let $\fei n$ be the manifold obtained by integer surgery with integer $n$ on the figure-8 knot, then $Z({\rm F8}_5, D(G_{5,11},u))$ distinguishes between the theories of different $u$.  
\end{proposition}
\begin{proof}
  The partition function is given by \cite{witten1989quantum}
  \begin{equation}
    Z({\rm F8}_n) = \frac{1}{D^2}\sum_i d_i\q_i^n K_i
    \label{eq: fig-8-partition-fn}
  \end{equation}
where $d_i,\q_i$ are the quantum dimensions and topological spins of the anyons, $D^2=\sum_i d_i^2$, and the sum is taken over all anyon species. The partition function is then calculated using the code presented in \cite{Delaney_Tran_2018}, giving 
\begin{equation}
  \begin{aligned}
    Z(\fei 5,u=0)&=\frac{53}{11} \\
    Z(\fei 5,u=1)&=-\frac{13}{11}+\frac{6\sqrt5}{5}-i\sqrt{2(5+\sqrt5)} \\
    Z(\fei 5,u=2)&=-\frac{13}{11}-\frac{6\sqrt5}{5}-i\sqrt{2(5-\sqrt5)} \\
    Z(\fei 5,u=3)&=-\frac{13}{11}-\frac{6\sqrt5}{5}+i\sqrt{2(5-\sqrt5)} \\
    Z(\fei 5,u=4)&=-\frac{13}{11}+\frac{6\sqrt5}{5}+i\sqrt{2(5+\sqrt5)} \\
  \end{aligned}
\end{equation}
\end{proof}

While the partition functions on $\fei n$ suffice to distinguish the theories $D(G_{5,11},u)$, they do not generally distinguish the theories $D(G_{p,q},u)$. We analyze this further in App. \ref{app: fei partition function}. Specifically, we prove the following:
\begin{proposition}
\label{prop: fei q conditions}
    For $p|q-1$, let $u, u'\neq 0$ and $u\neq u'$ such that both are either quadratic residues or quadratic nonresidues modulo $p$ ($u$ is a quadratic residue if $u=t^2$ for some $t$). Then:
    \begin{enumerate}
        \item $Z(\fei n,D(G_{p,q},u))=Z(\fei n,D(G_{p,q},u'))$ for any $p\nmid n$.
        \item  We have $Z(\fei p,D(G_{p,q},u))\neq Z(\fei p,D(G_{p,q},u'))$ if and only if there exists $k$ such that $s^{2k}-3s^k+1=0\mod q$. In particular, $Z(\fei p,D(G_{p,q},u))= Z(\fei p,D(G_{p,q},u'))$ for any $q\neq\pm1\mod5$.
    \end{enumerate}
\end{proposition}
This shows that there are infinitely many examples where $D(G_{p,q},u)$ are not distinguished by $\fei n$. In the appendix, we show that there are also infinitely many examples that are distinguishable in this manner. 

After showing that the manifolds $\fei n$ are useful for distinguishing beyond-modular data, it remains to show that there exist locally-achiral gems representing them. This will be done in the following sections. Note that, when $Z(\fei p)$ distinguishes the theories and $u'\neq -u$, the partition functions are generically distinguished by both the phase and the magnitude of $Z(\fei p)$. As a result, a locally-achiral gem is not needed in this case.

\begin{figure}
  \centering
  \includegraphics[scale=1.3]{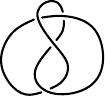}
  \caption{The figure-8 knot}\label{fig: fig-8}
\end{figure}

\subsection{Generalities of 3-gems}
In this section we recall several properties of a 3-manifold $M$ with a gem $\GG$. A thorough discussion of gems of 3-manifolds and their application to the classification of 3-manifolds with small gems, can be found in the book \cite{Lins_1995}. The first property is a simplification of Definition \ref{def: gem} in the case of 3-manifolds. Letting $\GG_{ij}$ be the $ij$ residue of $\GG$ (the graph obtained by removing the edges with colors $ij$ from $\GG$) and $g_{ij}$ being the number of disconnected components of $\GG_{ij}$, we have:
\begin{theorem}[\cite{Lins_1995}]
  \label{th: 3-gem-condition}
  A 4-graph $\GG$ with connected 3-residues $\GG_i$ (graphs obtained by removing one color $i$) is a gem if and only if 
  \begin{enumerate}
  \item $g_{ij}=g_{kl}$ for all $ijkl$ distinct.
  \item $g_{12}+g_{23}+g_{13}=2+\abs{\GG}/2$, where $\abs{\GG}$ is the number of vertices in $\GG$.
  \end{enumerate}
\end{theorem}
The theorem follows since $\GG_i$ are triangulations of 2-dimensional manifolds. These are spheres if and only if their Euler characteristic is 2, and the latter can be calculated combinatorially in terms of $\GG$. 

We will also use a method for extracting the fundamental group from a 3-gem. To do so, first consider the disconnected components of $\GG_{12}$, labeled by the symbols $x_1,...,x_{m+1}$. These are cycles of edge colors alternating between 3,4. The generators of $\pi_1(M)$ will be all symbols $x_i$ but one (any choice of $x_1,...,x_m$) will do. For the relations, consider the $m+1$ cycles of $\GG_{34}$ (their number is the same as the number of cycles in $\GG_{12}$ by Theorem \ref{th: 3-gem-condition}). Assuming that the $i$'th cycle passes through vertices $v_{k_1},...,v_{k_{2l}}$ and let $x_{k_i}$ be the cycle of $\GG_{12}$ containing $v_{k_i}$. The relation associated with the cycle $i$ is then 
\begin{equation}
  r_i=x_{k_1} x_{k_2}^{-1}x_{k_3}\cdots x_{k_{2l}}^{-1}.
\end{equation}
where it is implied that $x_{m+1}=1$. We have
\begin{theorem}[\cite{Lins_1995}]
  \label{th: fun-group-from-gem}
  Let $M$ be a 3-manifold with gem $\GG$. Its fundamental group is given by
  \begin{equation}
    \pi_1(M)=\left\langle x_1,...,x_m | r_1,...,r_m \right\rangle.
    \label{eq: pi_1(M)}
  \end{equation}
\end{theorem}
The usefulness of the fundamental group in our case lies in the fact that, in many cases, it characterizes the 3-manifold uniquely, based on the following theorem: 
\begin{theorem}[Theorem 2.2 of \cite{aschenbrenner20123}]
  \label{th: fun-group-classification}
  Let $M$,$N$ be two closed, orientable, prime manifolds, which are not lens spaces, with isomorphic fundamental groups $\pi_1(M)\cong\pi_1(N)$. Then $M\cong N$.
\end{theorem}
A manifold is prime if it cannot be decomposed as the connected sum of two manifolds $M=M_1 \# M_2$ with $M_1,M_2\not\cong S^3$. For closed orientable manifolds, the other conditions can be checked from the fundamental group: if $M=M_1\# M_2$ then its fundamental group is the free product $\pi_1(M)=\pi_1(M_1)*\pi_1(M_2)$ (both are nontrivial by the Poincar\'e conjecture). If $M$ is a lens space then $\pi_1(M)=\Z_n$ for some $n$, in particular, it is finite.

\subsection{Construction of locally-achiral 3-gems}
\begin{figure}
  \begin{center}
    \begin{tikzpicture}[manifold gem,every node/.append style={font={\tiny}}]
      \def\r{1}
      \def\rr{1.2}
      \def\d{1.5}
      \def\dd{1.5}
      \def\x{3.8}
      \def\y{2.2}
      \def\m{6}
      \def\mm{8}
      \def\mmm{10}

      \foreach \q in {1,3,...,\m}{
        \node[v1] (1-\q) at ({cos(360/6*(\q-2.5))*\r+\d},{sin(360/6*(\q-2.5))*\r}) {\q};
      }
      \foreach \q in {2,4,...,\m}{
        \node[v2] (1-\q) at ({cos(360/6*(\q-2.5))*\r+\d},{sin(360/6*(\q-2.5))*\r}) {\q};
      }
      \foreach \q in {1,3,...,\m}{
        \pgfmathtruncatemacro{\qq}{\q+1}
        \draw[e0] (1-\q) to (1-\qq);
        }
      \foreach \q in {3,5,...,\m}{
        \pgfmathtruncatemacro{\qq}{\q-1}
        \draw[e1] (1-\q) to (1-\qq);
        }
      \draw[e1] (1-1) to (1-\m);

      \foreach \q in {1,3,...,\mm}{
        \node[v1] (2-\q) at ({cos(360/8*(\q+.5))*\rr-\dd},{sin(360/8*(\q+.5))*\rr}) {\q};
      }
      \foreach \q in {2,4,...,\mm}{
        \node[v2] (2-\q) at ({cos(360/8*(\q+.5))*\rr-\dd},{sin(360/8*(\q+.5))*\rr}) {\q};
      }
      \foreach \q in {1,3,...,\mm}{
        \pgfmathtruncatemacro{\qq}{\q+1}
        \draw[e0] (2-\q) to (2-\qq);
        }
        \foreach \q in {3,5,...,\mm}{
        \pgfmathtruncatemacro{\qq}{\q-1}
        \draw[e1] (2-\q) to (2-\qq);
        }
      \draw[e1] (2-1) to (2-\mm);
      
      \foreach \q in {1,3,...,\mmm}{
        \node[v2] (3-\q) at ({cos(360/\mmm*(\q-2.5))*\x},{sin(360/\mmm*(\q-2.5))*\y}) {\q};
      }
      \foreach \q in {2,4,...,\mmm}{
        \node[v1] (3-\q) at ({cos(360/\mmm*(\q-2.5))*\x},{sin(360/\mmm*(\q-2.5))*\y}) {\q};
      }
      \foreach \q in {1,3,...,\mmm}{
        \pgfmathtruncatemacro{\qq}{\q+1}
        \draw[e0] (3-\q) to (3-\qq);
        }
        \foreach \q in {3,5,...,\mmm}{
        \pgfmathtruncatemacro{\qq}{\q-1}
        \draw[e1] (3-\q) to (3-\qq);
        }
      \draw[e1] (3-1) to (3-\mmm);

      \foreach \q in {1,...,4}
        \draw[e2] (1-\q) -- (3-\q);
      \foreach \q in {1,...,6}{
        \pgfmathtruncatemacro{\qq}{\q+4}
        \draw[e2] (2-\q) -- (3-\qq);
      }
      \foreach \q in {5,6}{
        \pgfmathtruncatemacro{\qq}{13-\q}
        \draw[e2] (1-\q) -- (2-\qq);
      }

      \draw[dashed] (-4.5,0) -- (4.5,0);

    \end{tikzpicture}
  \end{center}
  \caption{A gem of $S^2$ with three cycles in $\GG_{12}$ and $n_1,n_2,n_3=6,8,10$. The dashed line represents the reflection plane. }\label{fig: S2-gems-lemma}
\end{figure}
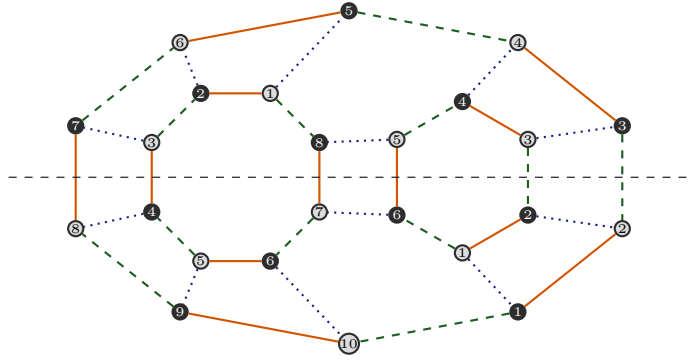

Here we present a general method for constructing locally-achiral gems of 3-manifolds. The method will allow for computer search over a space of locally-achiral gems of $N$ vertices. We will then use it to find locally-achiral gems for the manifolds $\fei n$. 

Our method is very similar to the one presented in \cite{Basak_2016}. The idea is to specialize to the case where $g_{12}=3$ (the case of $g_{12}=2$ covers only the Lens spaces). This case gives us more control and enables a simple way to verify that a gem is locally-achiral. The main tool will be the following lemma: 
\begin{lemma}
  \label{lem: s2-locally-achiral-condition}
  Let $\GG$ be a gem for $S^2$. Assume that $\GG_1$ has exactly three disconnected cycles and let $n_i$ $i=1,2,3$ denote the number of vertices in each cycle. Assume further that edges of color $1$ only connect between different cycles. Then $\GG$ is determined uniquely up to coloring. Furthermore, if $\abs{\GG}=4k$ for some integer $k$, then $\GG$ is reflection-positive.
\end{lemma}
\begin{proof}
  Let $m_{ij}$ be the number of color 1 edges connecting between cycles $i,j$. We have $n_1=m_{12}+m_{13},n_2=m_{23}+m_{12},n_3=m_{13}+m_{23}$, which is inverted to give 
  \begin{equation}
    \begin{aligned}
      m_{12} &= \frac{1}{2}(n_1+n_2-n_3), & m_{13} &= \frac{1}{2}(n_1+n_3-n_2) ,& m_{23} &= \frac{1}{2}(n_2+n_3-n_1).
    \end{aligned}
    \label{eq: s2-m-from-n}
  \end{equation}
  Since $\GG$ is a gem of $S^2$ it is the dual graph of a triangulation of $S^2$ and is thus a planar graph. Begin by drawing the three $\GG_1$ cycles on the plane and label the vertices in cycle $i$ (see Fig. \ref{fig: S2-gems-lemma}a). Since the graph is planar, for each pair $ij$ of cycles, if color $1$ edges connect $v_{i,s}$ and $v_{i,s+t}$ to cycle $j$, then either all of $v_{i,s+1},...,v_{i,s+t-1}$ are connected to cycle $j$, or all of $v_{i,1},...,v_{i,s-1},v_{i,s+t+1},...,v_{i,n_i}$. Furthermore, consecutive vertices of cycle $i$ should be connected to consecutive vertices of cycle $j$. We therefore get a unique way of connecting color 1 (up to relabeling): connect the vertices $v_{1,1},...,v_{1,m_{12}}$ to $v_{2,1}, ... v_{2,m_12}$, the vertices $v_{1,m_{12}+1},...,v_{1,n_1}$ to $v_{3,1},...,v_{3,m_{13}}$, and the vertices $v_{2,m_{12}+1},...,v_{2,n_2}$ to $v_{3,n_3},...,v_{3,m_{13}+1}$ (note the inverted order in the last sequence), see Fig. \ref{fig: S2-gems-lemma}c. This shows the uniqueness of $\GG$.

  If $\abs{\GG}=0\mod4$ then, since $n_i$ are even, by \eqref{eq: s2-m-from-n}, $m_i$ are also even (since $n_i+n_j-n_k=n_i+n_j+n_k\mod 4$). We can therefore make the cut on the edges connecting $(x_{1,m_{12}/2}$, $x_{1,m_{12}/2+1})$, $(x_{2,m_{12}/2},x_{2,m_{12}/2+1})$, $(x_{3,m_{13}/2}, x_{3,m_{13}/2+1})$, $(x_{1,n_1-m_{13}/2-1},x_{1,n_1-m_{12}/2})$, $(x_{2,n_2-m_{23}/2-1},x_{2,n_2-m_{23}/2})$, $(x_{3,n_3-m_{23}/2-1},x_{3,n_3-m_{23}/2})$.
\end{proof}

For 3-gems with $g_{12}=3$, since the space of possible $\GG_{1}$ is highly constrained by the previous lemma, for a given number of vertices, the space of possible 3-gems satisfying this condition has polynomially-bounded size. This suggests the following algorithm for listing locally-achiral 3-gems. Fix $n_1,n_2,n_3$ with $n_1+n_2+n_3=4k$, and draw the unique graph for $\GG_1$. We then want to add the color 1 edges to the graph (note that color-1 edges in the proof of the lemma become color 2 here). Color-1 edges should connect the three cycles of $\GG_{12}$. As above, they should connect consecutive vertices of cycle $i$ to consecutive edges of cycle $j$ to ensure that $\GG_2$ is planar. Therefore, taking three integers $q_1,q_2,q_3$, we connect the vertices $v_{1,1+q_1},...,v_{1,m_{12}+q_1}$ to $v_{2,1+q_2}, ... v_{2,m_12+q_2}$, the vertices $v_{1,m_{12}+1+q_1},...,v_{1,n_1}$ to $v_{3,1+q_3},...,v_{3,m_{13}+q_3}$, and the vertices $v_{2,m_{12}+1+q_2},...,v_{2,n_2+q_2}$ to $v_{3,n_3+q_3},...,v_{3,m_{13}+1+q_3}$ (the second index of $v_{i,j}$ should be taken mod $n_i$). The graph is bipartite if and only if $q_i$ are either all even or all odd. Furthermore, let 3 be the color of the edge connecting $v_{1,n_1},v_{1,1}$, we have $g_{23}=k-1,g_{24}=k$. To ensure $g_{14}=g_{23}, g_{13}=g_{24}$ we need to choose $q_i$ to be all odd. Finally, by Theorem \ref{th: 3-gem-condition}, the resulting graph is a gem if and only if $g_{34}=g_{12}=3$. This condition can be checked given a choice of $q_i$. Thus, for every choice of odd $q_i$ such that $g_{34}=3$ we get a valid gem. Furthermore, if in addition colors 3,4 connect distinct $\GG_{34}$ cycles, by lemma \ref{lem: s2-locally-achiral-condition} we are promised that this gem is locally achiral. 

Using the above procedure, we can now iterate over $n_1,n_2,n_3,q_1,q_2,q_3$ and list the resulting 3-manifolds. We implemented this procedure and compared the resulting 3-manifolds with those in the 3-manifold census in Regina \cite{regina}. The resulting list reveals a rich landscape of manifolds obtained this way and allowed for identifying certain examples of interest. Below, we show how the manifolds $\fei n$ for $n\ge 5$ can be obtained.

\subsection{Locally-achiral gems for \texorpdfstring{$\fei n$}{F8\_n}}
\begin{figure}
  \begin{center}
    \begin{tikzpicture}[manifold gem,every node/.append style={font={\tiny}}]
      \def\r{1.2}
      \def\rr{2.0}
      \def\d{3.0}
      \def\dd{2.5}
      \def\x{6.2}
      \def\y{3.8}
      \def\m{14}
      \def\mm{18}
      \def\mmm{20}

      \foreach \q in {1,3,...,\m}{
        \pgfmathtruncatemacro{\l}{(\q+1)/2}
        \node[v1] (1-\q) at ({cos(360/\m*(\q-4.5))*\r+\d},{sin(360/\m*(\q-4.5))*\r}) {\l};
      }
      \foreach \q in {2,4,...,\m}{
        \pgfmathtruncatemacro{\l}{mod(\q/2,\m/2)+1}
        \node[v2] (1-\q) at ({cos(360/\m*(\q-4.5))*\r+\d},{sin(360/\m*(\q-4.5))*\r}) {\l};
      }
      \foreach \q in {1,3,...,\m}{
        \pgfmathtruncatemacro{\qq}{\q+1}
        \draw[e0] (1-\q) to (1-\qq);
        }
      \foreach \q in {3,5,...,\m}{
        \pgfmathtruncatemacro{\qq}{\q-1}
        \draw[e1] (1-\q) to (1-\qq);
        }
      \draw[e1] (1-1) to (1-\m);

      \foreach \q in {1,3,...,\mm}{
        \pgfmathtruncatemacro{\l}{mod((\q-5)/2+\mm/2,\mm/2)+1}
        \node[v1] (2-\q) at ({cos(360/\mm*(\q+5.5))*\rr-\dd},{sin(360/\mm*(\q+5.5))*\rr}) {\l};
      }
      \foreach \q in {2,4,...,\mm}{
        \pgfmathtruncatemacro{\l}{mod(\q/2-2+\mm/2,\mm/2)+1}
        \node[v2] (2-\q) at ({cos(360/\mm*(\q+5.5))*\rr-\dd},{sin(360/\mm*(\q+5.5))*\rr}) {\l};
      }
      \foreach \q in {1,3,...,\mm}{
        \pgfmathtruncatemacro{\qq}{\q+1}
        \draw[e0] (2-\q) to (2-\qq);
        }
        \foreach \q in {3,5,...,\mm}{
        \pgfmathtruncatemacro{\qq}{\q-1}
        \draw[e1] (2-\q) to (2-\qq);
        }
      \draw[e1] (2-1) to (2-\mm);

      \foreach \q in {1,3,...,\mmm}{
        \pgfmathtruncatemacro{\l}{mod((\q-15)/2+\mmm/2,\mmm/2)+1}
        \node[v2] (3-\q) at ({cos(360/\mmm*(\q-2.5))*\x},{sin(360/\mmm*(\q-2.5))*\y}) {\l};
      }
      \foreach \q in {2,4,...,\mmm}{
        \pgfmathtruncatemacro{\l}{mod((\q-14)/2+\mmm/2,\mmm/2)+1}
        \node[v1] (3-\q) at ({cos(360/\mmm*(\q-2.5))*\x},{sin(360/\mmm*(\q-2.5))*\y}) {\l};
      }
      \foreach \q in {1,3,...,\mmm}{
        \pgfmathtruncatemacro{\qq}{\q+1}
        \draw[e0] (3-\q) to (3-\qq);
        }
        \foreach \q in {3,5,...,\mmm}{
        \pgfmathtruncatemacro{\qq}{\q-1}
        \draw[e1] (3-\q) to (3-\qq);
        }
      \draw[e1] (3-1) to (3-\mmm);

      \foreach \q in {1,...,8}
        \draw[e2] (1-\q) -- (3-\q);
      \foreach \q in {1,...,12}{
        \pgfmathtruncatemacro{\qq}{\q+8}
        \draw[e2] (2-\q) -- (3-\qq);
      }
      \foreach \q in {12,...,14}{
        \pgfmathtruncatemacro{\qq}{27-\q}
        \draw[e2] (1-\q) -- (2-\qq);
      }
      \draw[e2] (1-9) to[bend right=35] (2-18);
      \draw[e2] (1-10) to[bend right=20] (2-17);
      \draw[e2] (1-11) to[bend right=10] (2-16);
      \begin{scope}[every path/.style={e3}]
        \draw (3-14) to[out=-40, in=-120,looseness=.9] (1-14);
        \draw (3-15) to[out=-20, in=-120,looseness=.9] (1-1);
        \draw (3-16) to[out=-10, in=-120,looseness=.9] (1-2);
        \draw (3-17) to[out=-00, in=-110,looseness=.8] (1-3);
        \draw (3-18) to[out=15, in=-90,looseness=.9] (1-4);
        \draw (3-19) to[out=35, in=-70] (1-5);
        \draw (3-20) to[out=45, in=-45] (1-6);
        \draw (3-1) to[out=75, in=-20] (1-7);
        \draw (3-2) to[out=135, in=40,looseness=2.0] (2-10);
        \draw (3-3) to[out=140,in=50,looseness=1.8] (2-11);
        \draw (3-4) to[out=140,in=60,looseness=1.4] (2-12);
        \draw (3-5) to[bend right=55] (2-13);
        \draw (3-6) to[bend right=55] (2-14);
        \draw (3-7) to[bend right=55,looseness=1.2] (2-15);
        \draw (3-8) to[bend right=45] (2-16);
        \draw (3-9) to[bend right=25] (2-17);
        \draw (3-10) to (2-18);
        \draw (3-11) to (2-1);
        \draw (3-12) to (2-2);
        \draw (3-13) to (2-3);
        \draw (1-8) to[out=170,in=10,looseness=.9] (2-9);
        \draw (1-9) to[out=200,in=-10,looseness=1.2] (2-8);
        \draw (1-10) to[out=220,in=-20,looseness=1.4] (2-7);
        \draw (1-11) to[out=-130,in=-50,looseness=1.3] (2-6);
        \draw (1-12) to[out=-120,in=-70,looseness=1.25] (2-5);
        \draw (1-13) to[out=-120,in=-100,looseness=1.25] (2-4);
      \end{scope}
      \draw[very thick,darkblue,opacity=.6] (-3.2,1.7) -- (-5.7,2.1) -- (-6.5,0.2) -- (-4.0, 0.6) -- cycle;
      \node[darkblue,align=center,font={\scriptsize}] at (-6.1,2.4) {repeat $n-5$ \\times};

      \begin{scope}[shift={(7,0)},yscale=.7,xscale=1.3,every node/.append style={font={\Large}}]
      \draw[e1] (0,2) --++ (1,0) node[right] {$\Lambda$};
      \draw[e0] (0,1) --++ (1,0) node[right] {$A$};
      \draw[e2] (0,0) --++ (1,0) node[right] {$B$};
      \draw[e3] (0,-1) --++ (1,0) node[right] {$C$};
      \end{scope}
    \end{tikzpicture}
    
  \end{center}
  \caption{The gems for $\fei n$. The motif encircled in blue should be repeated $n-5$ times consecutively (here for $n=6$).}\label{fig: f8n-gems} \end{figure}
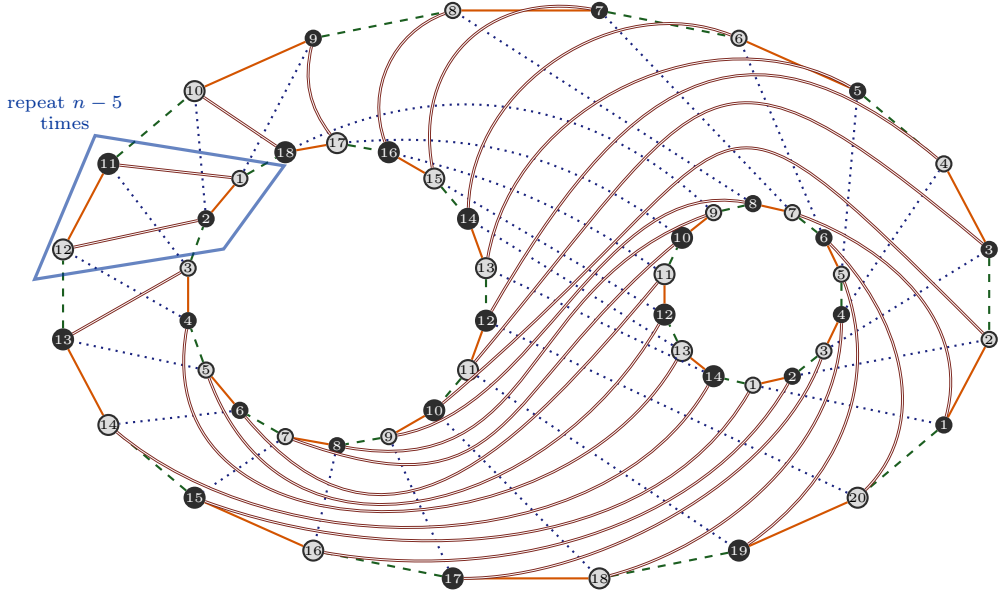

Through the search procedure above, we obtained the gems for $\fei n$, drawn in Fig. \ref{fig: f8n-gems}. A code implementing this search is available at \cite{sheffer2026gemscode}. They have $28+4n$ vertices and thus correspond to multi-entropy measures on $14+2n$ replicas. For $n=5$, we find the multi-entropy measure detecting the $G_{5,11}$ MS examples using 24 replicas. Since the complexity of evaluating $\arg\MM$ scales exponentially with the number of replicas, evaluating this entanglement measure remains many orders of magnitude beyond the reach of any practical method.

\begin{proposition}
  The gems drawn in Fig. \ref{fig: f8n-gems} correspond to the manifolds $\fei n$ for $n\ge 5$.
\end{proposition}
\begin{proof}
  \newcommand{\ii}{^{-1}}
  One checks that the 3-residues are connected, that $g_{12}=g_{34}=3$, and that the two remaining pairs of complementary residues satisfy $g_{ij}=g_{kl}$, with values $n+6$ and $n+7$. In particular, $g_{12}+g_{13}+g_{23}=2n+16=2+\abs{\GG}/2$, and by Theorem \ref{th: 3-gem-condition} these are indeed gems. Denote the resulting manifolds by $M_n$. Let us calculate the fundamental group using Theorem \ref{th: fun-group-from-gem}, choosing the generators to be two of the three cycles in $\GG_{12}$ (the inner circles in Fig. \ref{fig: f8n-gems}). We have 
  \begin{equation}
    \pi_1(M_n) = \left\langle x,y | y^{n-2} x\ii y x^2 y x\ii=1, y\ii x\ii y x y x y x\ii y\ii x=1\right\rangle. 
  \label{eq:pi1Mn}
  \end{equation}
  On the other hand, the fundamental groups of $\fei n$ are given by 
  \begin{equation}
  \label{eq:pi1F8n}
    \pi_1(\fei n) = \left\langle a,b| b\ii a\ii b a b\ii a b a\ii b\ii a = 1, a^n b a\ii b\ii a^2 b\ii a\ii b=1 \right\rangle .
  \end{equation}
  The first relation above is obtained from the presentation of the complement of the figure-8 knot in $S^3$ (see e.g. \cite{rolfsen2003knots}). The second relation is the result of the Dehn filling, requiring that $m^n l=1$ where $m,l$ are the elements of the knot fundamental group corresponding to the meridian and longitudinal lines of the boundary. 

  The two groups are isomorphic with the mapping $a=y, b=yx\ii y x y\ii$ (see App.~\ref{app:pi1}). By Thurston's theorem, $\fei n$ are hyperbolic (hence prime) for $n\ge5$. In particular, $\pi_1(\fei n)$ is not a free product (by Kneser's conjecture \cite{stallings1959kneser}) and is infinite. By Theorem \ref{th: fun-group-classification}, we have $M_n \cong \fei n$.
\end{proof}
An important note is that, as oriented manifolds, we cannot learn from the above proof whether $M_n\cong \fei n$ or $M_n\cong \overline{\fei n}$. In the context of detecting the MS examples, this would mean that the resulting multi-entropy measure could distinguish between $D(G_{5,11},u=1)$ and $D(G_{5,11},u=4)$, for example, but we could not tell a priori which is which. 

\section{Four-dimensional probes: obstructions to locally achiral manifolds and SPTs}
\label{sec: 4-manifolds}
\subsection{Vanishing of the Pontryagin numbers}
While the previous section presented a method of obtaining many locally-achiral gems, an important question that remains is whether any triangulated manifold can be represented in such a way. While the question remains completely open in three dimensions, we can say more in the four-dimensional case. 

To motivate the result below, we begin with a physical argument. In 3+1D, there exists a time-reversal symmetry protected topological order (T-SPT), whose action is given by the gravitational theta angle 
\begin{equation}
  Z(M)=e^{i\theta \int_M p_1}
  \label{eq: p_1}
\end{equation}
with $\theta=\pi$ and where
\begin{equation}
  \int_M p_1=-\frac{1}{8\pi^2}\int_M \mathrm{tr}\,\Omega\wedge\Omega\in \Z
  \label{eq: g-theta}
\end{equation}
is the Pontryagin number of $M$. Here $\Omega^\mu{}_\nu=\frac{1}{2}R^\mu{}_{\nu\rho\sigma}dx^\rho\wedge dx^\sigma$ is the curvature 2-form, where $R_{\mu\nu\rho\sigma}$ is the Riemann curvature tensor \cite{Nakahara_1990}. Importantly, in contrast with ``in cohomology" T-SPT phases, the action can be non-trivial on orientable manifolds \cite{Kapustin_2014}. For example, for the complex projective plane $\CC P^2$ we have $Z(\CC P^2)=-1$. The fact that it is a T-SPT follows since, in the absence of time-reversal symmetry, the variable $\theta$ can flow continuously, so we can adiabatically connect it to the trivial phase. In the presence of time-reversal symmetry we have $\theta=0,\pi$, and the latter represents then nontrivial SPT phase. A lattice model for this T-SPT is given by the ``Three-fermion Walker-Wang" (3FWW) model \cite{Walker_Wang_2012, Burnell_Chen_Fidkowski_Vishwanath_2014, Haah_Fidkowski_Hastings_2023}.

Assume now that there is a locally-achiral gem representing $\CC P^2$. The physical arguments in Sec. \ref{sec: multi-ent} imply that we could construct a multi-entropy measure, invariant under any ``generic" low-depth circuit (or adiabatic evolution), such that $\mathcal{M}(\y)\propto -1$, where $\ket\y$ is the 3FWW model. This would contradict the fact that, since it is a T-SPT model, it can be obtained by adiabatic evolution from the trivial state, for which $\mathcal{M}=1$ (this evolution will break time-reversal symmetry along the path). While the arguments above are not rigorous, they strongly suggest that $\CC P^2$ is not locally-achiral.

Beyond four dimensions, in any dimension $D=0\mod 4$, similar T-SPT phases could be defined, whose action is $\int_M p_{i_1}\cdots p_{i_k}\mod 2$ where $p_i$ are the Pontryagin classes, given as polynomials of degree $2i$ in $\Omega_{\mu\nu}$. The same argument implies the following conjecture:

\begin{conjecture}
  Let $M$ be a locally-achiral smooth manifold of dimension $D=0\mod4$. Then all Pontryagin numbers of $M$ vanish.
\end{conjecture}
While we are not able to prove this conjecture in full generality, we can show that it is true for $D=4$. Note that the SPT argument given above results in a slightly weaker form of the conjecture, namely that for any locally-achiral manifold the Pontryagin numbers should vanish mod 2.

\subsection{Four-dimensional manifolds}
\label{subsec: 4d manifolds}
\newcommand{\lk}{{\rm lk}}

We now prove the following: 
\begin{theorem}
  Let $M$ be a smooth 4-manifold with a locally-achiral gem. Then $\int_M p_1=0$.
  \label{th: vanishing-p1}
\end{theorem}
Below, we consider triangulated 4-manifolds. Since the gem of $M$ only gives a pseudo-triangulation (4-simplexes are not uniquely specified by their vertices), we need to consider its barycentric subdivision to obtain proper (oriented) triangulations $K$. For a vertex $v$ in $K$, define $\str(v)$ to be the union of simplexes adjacent to $v$, it is a topological 4-ball, and the link $\lk(v)$ to be the boundary of $\str(v)$. It is a triangulated 3-sphere. The proof relies on the following theorem:
\begin{theorem}[\cite{levitt1978existence}]
  There exists a function $\phi$ from oriented triangulated 3-spheres to $\mathbb{Q}$ such that, for any closed, oriented, triangulated smooth 4-manifold $M$, we have
\begin{equation}
    \int_M p_1 = \sum_v \phi(\lk(v)).
    \label{eq: p1-formula}
  \end{equation}
\end{theorem}
Note that, since $\int_M p_1=-\int_{\bar{M}} p_1$, we can choose $\phi$ to be odd under orientation reversal. Indeed, if it is not, it can be replaced by $\tilde{\phi}(s)=(\phi(s)-\phi(\bar{s}))/2$ which also satisfies \eqref{eq: p1-formula} and is odd under orientation-reversal. We can now prove Thm. \ref{th: vanishing-p1}:
\begin{proof}
    The vertices of $K$ are of two types. A vertex added by the subdivision is the barycenter $\hat\sigma$ of a simplex $\sigma$ of dimension $k\ge1$, and $\lk(\hat\sigma)$ is given as follows: its vertices are the barycenters of the proper faces of $\sigma$ and of the simplexes strictly containing $\sigma$. An odd permutation of the vertices of $\sigma$, acting on the barycenters of its faces and fixing all other vertices, is an orientation-reversing simplicial automorphism of $\lk(\hat\sigma)$, since it acts as a reflection on the factor $\partial\sigma\cong S^{k-1}$. Hence $\tilde\phi(\lk(\hat\sigma))=0$.
  
    The remaining vertices are the vertices of the gem. The vertices of color $i$ correspond to the connected components of the residue $\GG_i$, and the link of such a vertex $v$ is the barycentric subdivision of the triangulated 3-sphere encoded by its component. By local achirality, $\GG_i$ admits an involution that preserves edge colors and exchanges vertex colors. It permutes the components of $\GG_i$, and therefore defines an involution $\tau$ on the vertices of color $i$ with $\lk(\tau(v))\cong\overline{\lk(v)}$. The contributions of these vertices to \eqref{eq: p1-formula} therefore cancel, and we conclude that $\int_M p_1=0$.
\end{proof}

\section{Conclusions}
\label{sec: conclusions}
We showed that the notion of universal multipartite entanglement measures for topological order is intimately related to a topological condition of locally-achiral manifolds. We used this idea to construct universal entanglement measures that detect infinitely many of the MS examples. In higher dimensions, we argued that obstructions to local achirality are related to the existence of T-SPTs, and proved that, in dimension four, a non-zero Pontryagin number is an obstruction to a manifold being locally achiral. We now discuss a few open questions for future work.

Beyond nonzero Pontryagin numbers, we failed to find any obstructions to a manifold being locally-achiral. A rather bold conjecture, then, is that a manifold is locally-achiral if and only if all Pontryagin numbers vanish. This would likely imply, in particular, that all 3-manifolds are locally-achiral, and that all topological phases in 2+1d can be distinguished by universal multi-entropy measures. This assumes that any two theories can be distinguished by the partition function on some manifold $M$.

Beyond the general conjecture, it would be interesting to ask whether, for example, manifolds detecting the nontrivial invertible state in 4+1d \cite{chen2021exactly} could be locally achiral. This is a phase whose partition function is given by $e^{i\pi\int_M w_2w_3}$ where $w_i$ is the Stiefel-Whitney class, and is the lowest-dimensional example of nontrivial invertible states (without any symmetry) beyond chiral states in 2+1d. Since Stiefel-Whitney numbers on a triangulated manifold can be computed in polynomial time, it might be possible to automate a search procedure for a locally-achiral gem whose partition function detects this phase. Another possibility is to directly construct a locally-achiral gem of the Dold manifold $(\CC P^2\times S^1 )/\Z_2$, which is known to detect this phase. 

An interesting extension of the framework provided here is to fermionic topological phases. In that case, the manifold $M$ is required to be a spin manifold. It is interesting to ask how this requirement manifests in the form of allowed permutations. In addition, for fermionic systems, defining the permutation operators requires some choice of a sign structure, which should be related to the resulting spin structure of $M$. What will that mean for the cancellation of local contributions?

Finally, the major question of interest is to prove the main conjecture presented in Sec. \ref{sec: multi-ent}.  One question is whether it is possible to prove it under the assumption of the ``entanglement bootstrap" axioms \cite{Shi_Kato_Kim_2020}. Roughly speaking, a state satisfies the entanglement bootstrap axioms when the entanglement entropy of a region $A$ is given by $S(A)=\alpha \abs{\partial A}-\gamma$, for some constants $\alpha,\gamma$, with no subleading corrections. Notice that the example of \cite{flammia_topological_2009} shows that this property can hold exactly for the von-Neumann entropy while being true only approximately for any R\'enyi quantity (in particular, for multi-entropy measures).

\paragraph{Comment added.}
While preparing the resubmission of this manuscript, I became aware of \cite{delzotto2026multi}, whose results significantly strengthen the ones presented here. In particular, the authors show that the nontorsion part of the oriented cobordism group is the obstruction for universally obtaining $Z(M)$, and present a construction that extracts $Z(M)$ for any null-bordant $M$. 

\paragraph{Acknowledgements.}
I thank Ruihua Fan, Ady Stern, Erez Berg and Shinsei Ryu for inspiring discussions and collaboration on related projects. I thank Ramanjit Sohal, Michael Levin, Ryan Thorngren and especially Julian Gass for useful conversations, and Shinsei Ryu for insightful comments on the manuscript. Much of this project was completed during an unexpected period of traveling. I am grateful to Raquel Queiroz, Shinsei Ryu, and Lei Gioia for their hospitality during that time. I thank an anonymous referee for many useful suggestions, including the proofs in Sec. \ref{subsec: 4d manifolds} and App. \ref{app: inf solutions}. This work was supported by grants from the DFG (CRC/Transregio 183, EI 519/71, projects C02 and C03), and the ISF Quantum Science and Technology program (Grant no. 2478/24). I am supported by the Adams Fellowship Program of the Israel Academy of Sciences and Humanities.

\appendix
\section{Further details on the manifolds $\fei n$}
\subsection{Explicit form of the permutations}
Here we provide an explicit form of the permutations for the manifolds $\fei n$, as presented in Fig. \ref{fig: f8n-gems}. Write the replica indices as $(c,t)$ where $c=1,2,3$ is the cycle index, and $t$ is the index within the cycle, with $t$ in the range 1 to $l_c=7,n+3,n+4$ for the three cycles. Set $\pi_\Lambda$ to be the trivial permutation. The $A$ permutation is given by
\begin{equation}
  \pi_A :\;
    (1,t)\mapsto(1,t+1),\quad
    (2,t)\mapsto(2,t+1),\quad
    (3,t)\mapsto(3,t-1)
    \qquad (\text{$t$ mod $l_c$}).
\end{equation}
The next two permutations are more complicated, we write them in cycle form as
\begin{align}
  \pi_B &=
    \big((1,1),(3,4),(2,5)\big)\,
    \big((1,5),(2,8),(3,8)\big)
    \prod_{t=2}^{4}\big((1,t),(3,t+3)\big)
    \prod_{t=6}^{7}\big((1,t),(2,13-t)\big) \nonumber\\
    &\qquad\times
    \big((2,1),(3,n+4)\big)
    \prod_{t=2}^{4}\big((2,t),(3,t-1)\big)
    \prod_{t=9}^{n+3}\big((2,t),(3,t)\big), \\
  \pi_C &=
    \prod_{t=1}^{4}\big((1,t),(3,t)\big)
    \prod_{t=5}^{7}\big((1,t),(2,8-t)\big)
    \prod_{t=4}^{n+3}\big((2,t),(3,t+1)\big).
\end{align}

Next, we check that the 4-graph we obtained is indeed a locally-achiral gem. We first need to check that $\GG_{34}$ has 3 disjoint cycles. This follows since $\sigma=\pi_C^{-1}\pi _B$ has three cycles, given by 
\begin{equation}
\begin{aligned}
  \gamma_1 &= \big((1,1),(1,4),(2,6),(2,1),(2,n+3),(2,n+2),\dots,(2,8),(2,7),(2,2)\big),\\
  \gamma_2 &= \big((1,2),(2,4),(1,3),(2,5),(3,1),(1,6),(3,8),(2,3)\big),\\
  \gamma_3 &= \big((1,5),(3,9),(3,10),\dots,(3,n+4),(1,7),(3,7),(3,4),(3,6),(3,3),(3,5),(3,2)\big).
\end{aligned}
\end{equation}
To use lemma \ref{lem: s2-locally-achiral-condition} we need to check that colors 3,4 (corresponding to the regions $\Lambda$ and $A$) connect different cycles of $\GG_{34}$. This follows by explicitly checking that $\pi_B,\pi_A^{-1}\pi_B$ connect different cycles of $\sigma$.

\subsection{Isomorphism of the fundamental groups}
\label{app:pi1}

Here we show that the groups in \eqref{eq:pi1Mn} and \eqref{eq:pi1F8n} are isomorphic, via
\begin{equation}
  a = y, \qquad b = y x^{-1} y x y^{-1}.
  \label{eq:iso}
\end{equation}
The inverse is given by
\begin{equation}
  y = a, \qquad x = a^{n-2}\, b\, a^{-1} b^{-1} a^{2}.
  \label{eq:iso-inv}
\end{equation}
Denote the relators of \eqref{eq:pi1Mn} by
\begin{equation}
  R_1 = y^{n-2}x^{-1}yx^{2}yx^{-1}, \qquad
  R_2 = y^{-1}x^{-1}yxyxyx^{-1}y^{-1}x,
\end{equation}
and those of \eqref{eq:pi1F8n} by
\begin{equation}
  S_1 = b^{-1}a^{-1}bab^{-1}aba^{-1}b^{-1}a, \qquad
  S_2 = a^{n}ba^{-1}b^{-1}a^{2}b^{-1}a^{-1}b .
\end{equation}
The proof is a sequence of Tietze transformations.

\paragraph{Step 1: eliminating $x$.}
Let $c = x^{-1}yx$ and $d = xyx^{-1}$. As words in the free group,
\begin{equation}
  R_1 = y^{n-2}\, c\, d, \qquad R_2 = y^{-1}\, c\, y\, d\, y^{-1} x .
\end{equation}
The relation $R_1 = 1$ is therefore equivalent to $d = c^{-1}y^{2-n}$. Given this, $R_2 = 1$ is equivalent to
\begin{equation}
  x = y\, d^{-1} y^{-1} c^{-1} y = y^{n-1}\, c\, y^{-1} c^{-1}\, y \equiv w(y,c).
  \label{eq:x-of-yc}
\end{equation}
Hence $\pi_1(M_n)$ is generated by $y$ and $c$. Using \eqref{eq:x-of-yc} to eliminate $x$, we obtain
\begin{equation}
  \pi_1(M_n) = \left\langle y, c \;\middle|\; w\,c\,w^{-1} = y,\;\; w\,y\,w^{-1} = c^{-1}y^{2-n} \right\rangle .
  \label{eq:pi1-yc}
\end{equation}
Here the first relation is the definition $c = x^{-1}yx$, and the second is $R_1 = 1$ written as $xyx^{-1} = c^{-1}y^{2-n}$.

\paragraph{Step 2: change of generators.}
Set $a = y$ and $b = y c y^{-1}$, so that $c = a^{-1}ba$; this reproduces \eqref{eq:iso}. Then
\begin{equation}
  w = a^{n-1}\,(a^{-1}ba)\,a^{-1}\,(a^{-1}b^{-1}a)\,a = a^{n-2}\, b\, a^{-1} b^{-1} a^{2},
\end{equation}
which, by \eqref{eq:x-of-yc}, gives \eqref{eq:iso-inv}. Define
\begin{equation}
  P = b a^{-1} b^{-1} a b, \qquad Q = a b a^{-1} b^{-1} a .
\end{equation}
A direct computation gives
\begin{equation}
  w\, c\, w^{-1} a^{-1} = a^{n-2}\, \big(P Q^{-1}\big)\, a^{-(n-2)} ,
\end{equation}
so the first relation in \eqref{eq:pi1-yc} is equivalent to $P = Q$. Since $S_1 = P^{-1} Q$, this relation is equivalent to $S_1 = 1$. For the second relation in \eqref{eq:pi1-yc}, we find
\begin{equation}
  a\left[\, w\, a\, w^{-1} \left(c^{-1} a^{2-n}\right)^{-1} \right] a^{-1}
  = a^{n-1}\, P\, a\, b^{-1} a^{-1} b .
\end{equation}
Using $P = Q$, the right-hand side becomes
\begin{equation}
  a^{n-1}\, Q\, a\, b^{-1} a^{-1} b = a^{n}\, b a^{-1} b^{-1} a^{2} b^{-1} a^{-1} b = S_2 .
\end{equation}
The relations of \eqref{eq:pi1-yc} are thus equivalent to $S_1 = S_2 = 1$, and
\begin{equation}
  \pi_1(M_n) \cong \langle a, b \mid S_1, S_2 \rangle = \pi_1(F8_n).
\end{equation}

\section{The partition functions of \texorpdfstring{$\fei n$}{F8\_n}}
\label{app: fei partition function}
\newcommand{\dpq}{D(G_{p,q},u)}
Here we describe in more detail the partition functions of the manifolds $\fei n$ for the topological theories $D(G_{p,q},u)$. We will use the tools developed in \cite{Bonderson_2018_beyond} to calculate the figure-8 knot diagrams, $K_i$. The partition functions are then obtained using eq. \eqref{eq: fig-8-partition-fn}.

\subsection{Anyons in the \texorpdfstring{$D(G_{p,q},u)$}{D(G\_pq,u)} theories}
We begin by describing the anyons in the $\dpq$ theories, following the presentation in \cite{Bonderson_2018_beyond}. Let $a,b$ be the generators of $G_{p,q}$ as in \eqref{eq:G_pq}. The anyons in $\dpq$ are given by tuples $([g],\pi_g)$ where $g$ is a conjugacy class of an element $g$ in $G_{p,q}$, and $\pi_g$ is an irreducible projective representation of the centralizer of $g$ in $G_{p,q}$. The conjugacy classes of $G_{p,q}$ are $\{e\}$; the $(q-1)/p$ classes $[a^k]=\{a^{ks^j}\}_{j}$, where $k$ runs over representatives of the cosets of $\langle s\rangle$ in $\mathbb{F}_q^\times$; and the $p-1$ classes $[b^k]$, $k=1,\dots,p-1$. 

To describe the projective representations we use the form of the 3-cocycles of $H^3(G_{p,q},U(1))$ given by
\begin{equation}
    \omega_u([b^m],[b^n],[b^t])=\begin{cases}
    1, & \textrm{if }m+n<p\\
    e^{\frac{2\pi i}{p}u t}, & \textrm {if }m+n\ge p\\
    \end{cases}.
\end{equation}
for elements of the conjugacy classes of $[b^k]$, and 1 otherwise. It induces a 2-cochain given by
\begin{equation}
    \theta_g(x,y)=\frac{\omega(g,x,y)\omega(x,y,(xy)^{-1}gxy)}{\omega(x,x^{-1}gx,y)}.
\end{equation}
The projective representations $\pi_g$ then satisfy
\begin{equation}
    \pi_g(xy)=\pi_g(x)\pi_g(y)\theta_g(x,y).
\end{equation}
Denote by $B_{k,n}$ the anyons of the form $([b^k],\pi_k^n)$ where $\pi_k^n$ is one of the projective representations of the centralizer $G_{b^k}\cong \Z_p$. It is given by 
\begin{equation}
    \pi_k^n(b^l)=e^{\frac{2\pi i }{p^2}(np+uk)l}.
\end{equation}
with $n=0,...,p-1$. The $B$-type anyons are the ones that are of interest to us. This is because, for all other anyons, the associated projective representations are just regular representations, and therefore both $K_i,d_i,\theta_i$ are $u$ independent. We therefore have 
\begin{equation}
    Z(\fei n,u)-Z(\fei n,u')=\frac{1}{D^2}\sum_{k,t}\qty[d_{B_{k,t}}\theta_{B_{k,t}}^n(u) K_{B_{k,t}}(u)
    -d_{B_{k,t}}\theta_{B_{k,t}}^n(u') K_{B_{k,t}}(u')].
\end{equation}
The quantum dimensions and topological twists are given in \cite{Bonderson_2018_beyond} as
\begin{align}
    d_{B_{k,t}}&=q ,\\
    D&=pq, \\
    \theta_{B_{k,t}}&=e^{\frac{2\pi i}{p}kt}e^{\frac{2\pi i}{p^2}k^2u}.
\end{align}
We next calculate $K_{B_{k,t}}$.

\subsection{Calculation of the figure-8 knot diagrams}
\newcommand{\K}{\mathcal{K}}
\begin{figure}
    \centering
    \begin{tikzpicture}
        \node at (0,0){
        \begin{tikzpicture}[every node/.append style={font={\small}}]
            \node at (0,0) {\includegraphics[scale=1.3]{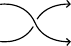}};
            \node at (-.6,.6) {$x$};
            \node at (.6,-.6) {$x$};
            \node at (.6,.6) {$y$};
            \node at (-.6,-.6) {$x\rhd y$};
            \begin{scope}[shift={(0,-1.9)}]
            \node[yscale=-1] at (0,0) {\includegraphics[scale=1.3]{coloring.pdf}};
            \node at (-.6,.6) {$y$};
            \node at (.6,-.6) {$x\rhd y$};
            \node at (.6,.6) {$x$};
            \node at (-.6,-.6) {$x$};
            \end{scope}
        \end{tikzpicture}
        };
        \node at (6,0){
        \begin{tikzpicture}[every node/.append style={font={\scriptsize}}]
            \node at (0,0) {\includegraphics[scale=1.3]{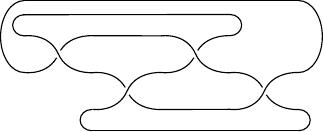}};
            \begin{scope}[shift={(-2.2,.25)}]
            \node at (-.6,.55) {$x$};
            \node at (.6,-.55) {$x$};
            \node at (.6,.55) {$y$};
            \node at (-.6,-.55) {$x\rhd y$};
            \end{scope}
            \begin{scope}[shift={(+0.7,.25)}]
            \node at (-.6,.55) {$y$};
            \node at (.6,-.55) {$y$};
            \node at (.6,.55) {$x$};
            \node at (-.6,-.55) {$y\rhd x$};
            \end{scope}
            \begin{scope}[shift={(-0.5,-.55)}]
            \node at (.7,-.55) {$(y\rhd x)\rhd x$};
            \node at (-.7,-.55) {$y\rhd x$};
            \end{scope}
            \begin{scope}[shift={(1.9,-.55)}]
            \node at (.85,-.65) {$(x\rhd y)\rhd y$};
            \node at (1.1,+.20) {$x\rhd y$};
            \end{scope}
        \end{tikzpicture}
        };
        \node at (0,-2.0) {(a)};
        \node at (6,-2.0) {(b)};
    \end{tikzpicture}
    \caption{(a) coloring rule for the crossing. (b) Coloring of the figure-8 knot.}
    \label{fig: quandles}
\end{figure}

In \cite{Bonderson_2018_beyond} it was shown that anyon diagrams for $B_{k,t}$ can be calculated in terms of ``quandles". We give their results without proof. A quandle $X_k$ is a set $X$ with a binary operation $\rhd:X\times X\to X$ satisfying certain compatibility conditions. Here we consider the Alexander quandle associated with $B_{k,t}$, where $X_k=\Z_q$ and 
\begin{equation}
    x\rhd y=(1-s^k)x+s^ky.
\end{equation}
The calculation of a knot diagram $L$ for $B_{k,t}$ is then carried as follows: Consider a braid diagram for $L$, and color the lines of each crossing by elements of $X_k$ as in Fig. \ref{fig: quandles}. The invariant of the knot $L$ with the anyon $B_{k,t}$ is then given by 
\begin{equation}
    \omega^{Wr(L)}C_{X_k}(L)
\end{equation}
where $Wr(L)$ is the writhe of $L$ (the number of over crossings minus the number of under crossings), $\omega=e^{\frac{2\pi i}{p^2}(sp+uk)k}$, and $C_{X_k}(L)$ is the number of consistent coloring of $L$ with $X_k$. For the figure-8 knot, the writhe is 0. We follow the diagram in Fig. \ref{fig: quandles}b and notice that consistent colorings are labeled by $x,y\in \Z_q$ satisfying
\begin{equation}
\begin{aligned}
    (y\rhd x)\rhd x&=x\rhd y, \\
    (x\rhd y)\rhd y&= y\rhd x,
    \end{aligned}
\end{equation}
which gives
\begin{equation}
    (1-3s^k+s^{2k})(x-y)\equiv0.
\end{equation}
where we use the equivalent sign for equality mod $q$. We therefore find that
\begin{equation}
    K_k\coloneq K_{B_{k,t}}=\begin{cases}
        q^2 & s^{2k}-3s^k+1\equiv0\\
        q & \rm{otherwise}
    \end{cases}.
\end{equation}
Let $\mathcal{K}$ be the space of solutions of $s^{2k}-3s^k+1=0\mod q$. Since $s^k\equiv x$ has at most 1 solution for any $x,s$, $\abs{\K}\le 2$. Furthermore, $x^2-3x+1\equiv0$ has a solution if and only if 5 is a quadratic residue mod $q$, and this is the case if and only if $q=\pm1\mod 5$. Also, since $\langle s\rangle$ is the unique subgroup of order $p$ in $\mathbb{F}_q^\times$, $s^k=x$ for some $k$ if and only if $x^p\equiv 1$. Finally, if $k$ is a solution then $-k$ is also a solution, since if $s^k=\frac{3\pm\sqrt5}{2}$ then $s^{-k}=\frac{3\mp\sqrt5}{2}$.

We therefore find that $\abs{\K}=2$ if $q=\pm1\mod 5$ and the solution of $x^2-3x+1\equiv0$ satisfies $x^p\equiv 1$, and $\abs{\K}=0$ otherwise. The partition function difference is then given by 
\begin{equation}
\begin{aligned}
    Z(\fei n,u)-Z(\fei n,u')&=\frac{1}{p^2}\sum_{k,t}\qty[\theta_{B_{k,t}}^n(u)-\theta_{B_{k,t}}^n(u')]+\frac{q-1}{p^2}\sum_{t,k\in\K}[\theta_{B_{k,t}}^n(u)-\theta_{B_{k,t}}^n(u')]. \\
    &=\frac{\delta_{n=0\mod p}}{p}\qty{\qty[\qty(\frac{u}{p})-\qty(\frac{u'}{p})]\zeta_{n,p}+(q-1)\sum_{k\in\K}\qty(e^{\frac{2\pi inu}{p^2}k^2}-e^{\frac{2\pi inu'}{p^2}k^2})}.
\end{aligned}
\end{equation}
Here $\qty(\frac{u}{p})$ is the Legendre symbol and $\zeta_{n,p}=\sum_{j=1}^pe^{2\pi i n j^2/p^2}$ is the quadratic Gauss sum. Let us assume $\qty(\frac{u}{p})=\qty(\frac{u'}{p})$ for simplicity, since otherwise the theories can be distinguished by the partition function of $L(p,1)$. Also, take $n=p$. If $\abs{\K}=0$ then $Z(\fei n,u)=Z(\fei n,u')$, otherwise
\begin{equation}
    Z(\fei n,u)-Z(\fei n,u')=2\frac{q-1}{p}(e^{\frac{2\pi i u}{p}k^2}-e^{\frac{2\pi i u'}{p}k^2}).
\end{equation}
which is zero if and only if $u=u'$. This proves Proposition \ref{prop: fei q conditions}.

\subsection{Infinitely-many solutions}
\label{app: inf solutions}
It is interesting to ask whether there are infinitely many pairs $(p,q)$ for which $Z(\fei p,u)$ distinguishes between the theories $\mathcal{D}(G_{p,q},u)$ with equal modular data. We now show that this is the case. Let $p\ge5$ be a prime, and let $q$ be a prime divisor of the Lucas number $L_p$. Here $L_n=\varphi^n+(-1)^n\varphi^{-n}$ with $\varphi=(1+\sqrt5)/2$; equivalently, $L_0=2$, $L_1=1$, and $L_{n+1}=L_n+L_{n-1}$. Similarly, let $F_n=(\varphi^n-(-1)^n\varphi^{-n})/\sqrt5\in\mathbb{Z}$ be the Fibonacci numbers. We claim that $(p,q)$ satisfies the conditions of Prop.~\ref{prop: fei q conditions}.

First, we notice that $q\neq 2,5$: using the recurrence relation, it can be checked that $L_n$ is even only if $3\mid n$, and $L_n\not\equiv0 \bmod 5$ for all $n$.  Next, the identity
\begin{equation}
    L_n^2-5F_n^2=4(-1)^n
\end{equation}
gives $5F_p^2\equiv 4 \bmod q$ for odd $p$. In particular $F_p\not\equiv0$, and $5\equiv(2F_p^{-1})^2 \bmod q$. Thus $5$ is a quadratic residue mod $q$. Taking $r$ with $r^2\equiv5$, the element $\phi=(1+r)/2\in\mathbb{F}_q$ satisfies $\phi^2=\phi+1$. Hence $x=\phi^2$ satisfies $x^2=3x-1$, i.e., $x^2-3x+1\equiv0 \bmod q$.

Next, the sequence $\phi^n+(1-\phi)^n=\phi^n+(-1)^n\phi^{-n}$ satisfies the Lucas recurrence with the same initial values (using $\phi(1-\phi)=-1$), so $L_n\equiv\phi^n+(-1)^n\phi^{-n} \bmod q$. For odd $p$, this gives
\begin{equation}
    0\equiv L_p\equiv\phi^p-\phi^{-p} \bmod q \quad\Rightarrow\quad x^p=\phi^{2p}\equiv1 \bmod q .
\end{equation}
Since $x\neq1$, $x$ has order exactly $p$ in $\mathbb{F}_q^\times$. It follows that $p\mid q-1$, so $G_{p,q}$ exists, and that $x=s^k$ for some $k$.

Finally, every prime $p\ge5$ yields at least one such $q$. This gives infinitely many examples; the first few are $(p,q)=(5,11)$, $(7,29)$, $(11,199)$, $(13,521)$.


\printbibliography

@article{dong_topological_2008,
  title = {Topological {Entanglement} {Entropy} in {Chern}-{Simons} {Theories}
           and {Quantum} {Hall} {Fluids}},
  volume = {2008},
  issn = {1029-8479},
  url = {http://arxiv.org/abs/0802.3231},
  doi = {10.1088/1126-6708/2008/05/016},
  number = {05},
  urldate = {2024-07-31},
  journal = {J. High Energy Phys.},
  author = {Dong, Shiying and Fradkin, Eduardo and Leigh, Robert G. and Nowling,
            Sean},
  month = may,
  year = {2008},
  note = {arXiv:0802.3231 [cond-mat, physics:hep-th]},
  pages = {016--016},
}

@article{kim_chiral_2022,
  title = {Chiral central charge from a single bulk wave function},
  volume = {128},
  issn = {0031-9007, 1079-7114},
  url = {http://arxiv.org/abs/2110.06932},
  doi = {10.1103/PhysRevLett.128.176402},
  number = {17},
  urldate = {2024-08-03},
  journal = {Phys. Rev. Lett.},
  author = {Kim, Isaac H. and Shi, Bowen and Kato, Kohtaro and Albert, Victor V.
            },
  month = apr,
  year = {2022},
  note = {arXiv:2110.06932 [cond-mat, physics:hep-th, physics:quant-ph]},
  pages = {176402},
}

@article{flammia_topological_2009,
  title = {Topological {Entanglement} {Renyi} {Entropy} and {Reduced} {Density}
           {Matrix} {Structure}},
  volume = {103},
  issn = {0031-9007, 1079-7114},
  url = {http://arxiv.org/abs/0909.3305},
  doi = {10.1103/PhysRevLett.103.261601},
  number = {26},
  urldate = {2024-09-12},
  journal = {Phys. Rev. Lett.},
  author = {Flammia, Steven T. and Hamma, Alioscia and Hughes, Taylor L. and Wen
            , Xiao-Gang},
  month = dec,
  year = {2009},
  note = {arXiv:0909.3305 [cond-mat, physics:quant-ph]},
  pages = {261601},
}

@article{penington_fun_2023,
  title = {Fun with replicas: tripartitions in tensor networks and gravity},
  volume = {2023},
  issn = {1029-8479},
  shorttitle = {Fun with replicas},
  url = {http://arxiv.org/abs/2211.16045},
  doi = {10.1007/JHEP05(2023)008},
  number = {5},
  urldate = {2024-10-14},
  journal = {J. High Energ. Phys.},
  author = {Penington, Geoff and Walter, Michael and Witteveen, Freek},
  month = may,
  year = {2023},
  note = {arXiv:2211.16045 [hep-th, physics:quant-ph]},
  pages = {8},
}

@misc{liu_multi_2024,
  title = {Multiwavefunction overlap and multientropy for topological ground states in ($2+1$) dimensions},
  author = {Liu, Bowei and Zhang, Junjia and Ohyama, Shuhei and Kusuki, Yuya and Ryu, Shinsei},
  journal = {Phys. Rev. B},
  volume = {112},
  issue = {12},
  pages = {125160},
  numpages = {16},
  year = {2025},
  month = {Sep},
  publisher = {American Physical Society},
  doi = {10.1103/tjcf-yryh},
  url = {https://link.aps.org/doi/10.1103/tjcf-yryh}
}

@book{rolfsen2003knots,
  title = {Knots and links},
  author = {Rolfsen, Dale},
  number = {346},
  year = {2003},
  publisher = {American Mathematical Soc.},
}

@article{witten1989quantum,
  title = {Quantum field theory and the Jones polynomial},
  author = {Witten, Edward},
  journal = {Communications in Mathematical Physics},
  volume = {121},
  number = {3},
  pages = {351--399},
  year = {1989},
  publisher = {Springer},
}

@article{kitaev_2006_topological,
  title = {Topological Entanglement Entropy},
  volume = {96},
  ISSN = {0031-9007, 1079-7114},
  DOI = {10.1103/PhysRevLett.96.110404},
  number = {11},
  journal = {Phys. Rev. Lett.},
  author = {Kitaev, Alexei and Preskill, John},
  year = {2006},
  month = mar,
  pages = {110404},
}

@article{Levin_Wen_2006,
  title = {Detecting topological order in a ground state wave function},
  volume = {96},
  ISSN = {0031-9007, 1079-7114},
  DOI = {10.1103/PhysRevLett.96.110405},
  abstractNote = {A large class of topological orders can be understood and
                  classified using the string-net condensation picture. These
                  topological orders can be characterized by a set of data (N,
                  d_i, F^{ijk}_{lmn}, delta_{ijk}). We describe a way to detect
                  this kind of topological order using only the ground state wave
                  function. The method involves computing a quantity called the
                  ``topological entropy’’ which directly measures the quantum
                  dimension D = sum_i d^2_i.},
  note = {arXiv:cond-mat/0510613},
  number = {11},
  journal = {Phys. Rev. Lett.},
  author = {Levin, Michael and Wen, Xiao-Gang},
  year = {2006},
  month = mar,
  pages = {110405},
}

@article{Shi_Kato_Kim_2020,
  title = {Fusion rules from entanglement},
  volume = {418},
  ISSN = {00034916},
  DOI = {10.1016/j.aop.2020.168164},
  note = {arXiv:1906.09376 [cond-mat, physics:quant-ph]},
  journal = {Annals of Physics},
  author = {Shi, Bowen and Kato, Kohtaro and Kim, Isaac H.},
  year = {2020},
  month = jul,
  pages = {168164},
}

@article{Liu_2024,
  title = {Anyon quantum dimensions from an arbitrary ground state wave function
           },
  volume = {15},
  ISSN = {2041-1723},
  DOI = {10.1038/s41467-024-47856-7},
  number = {1},
  journal = {Nat Commun},
  author = {Liu, Shang},
  year = {2024},
  month = jun,
  pages = {5134},
  language = {en},
}

@article{Gadde_Krishna_Sharma_2022,
  title = {A new multi-partite entanglement measure and its holographic dual},
  volume = {106},
  ISSN = {2470-0010, 2470-0029},
  DOI = {10.1103/PhysRevD.106.126001},
  abstractNote = {In this letter we define a natural generalization of the von
                  Neumann entropy to multiple parties that is symmetric with
                  respect to all the parties. We call this measure multi-entropy.
                  We show that for conformal field theories with holographic
                  duals, the multi-entropy is computed by the area of an
                  appropriate “soap-film” anchored on the boundary. We conjecture
                  the quantum version of this prescription that takes into
                  account the sub-leading corrections in G_N.},
  note = {arXiv:2206.09723 [hep-th]},
  number = {12},
  journal = {Phys. Rev. D},
  author = {Gadde, Abhijit and Krishna, Vineeth and Sharma, Trakshu},
  year = {2022},
  month = dec,
  pages = {126001},
  language = {en},
}

@article{williamson2016spurious,
  title = {Spurious Topological Entanglement Entropy from Subsystem Symmetries},
  author = {Williamson, Dominic J. and Dua, Arpit and Cheng, Meng},
  journal = {Phys. Rev. Lett.},
  volume = {122},
  issue = {14},
  pages = {140506},
  numpages = {6},
  year = {2019},
  month = {Apr},
  publisher = {American Physical Society},
  doi = {10.1103/PhysRevLett.122.140506},
  url = {https://link.aps.org/doi/10.1103/PhysRevLett.122.140506},
}

@article{Zou_Haah_2016,
  title = {Spurious long-range entanglement and replica correlation length},
  volume = {94},
  DOI = {10.1103/PhysRevB.94.075151},
  number = {7},
  journal = {Phys. Rev. B},
  publisher = {American Physical Society},
  author = {Zou, Liujun and Haah, Jeongwan},
  year = {2016},
  month = aug,
  pages = {075151},
}

@article{Gass_Levin_2024,
  title = {Many-body systems with spurious modular commutators},
  url = {http://arxiv.org/abs/2405.15892},
  DOI = {10.48550/arXiv.2405.15892},
  note = {arXiv:2405.15892 [quant-ph]},
  number = {arXiv:2405.15892},
  publisher = {arXiv},
  author = {Gass, Julian and Levin, Michael},
  journal = "",
  year = {2024},
  month = may,
}

@article{Mignard_Schauenburg_2021,
  journal = "",
  title = {Modular categories are not determined by their modular data},
  url = {http://arxiv.org/abs/1708.02796},
  DOI = {10.48550/arXiv.1708.02796},
  note = {arXiv:1708.02796 [math]},
  number = {arXiv:1708.02796},
  publisher = {arXiv},
  author = {Mignard, Michaël and Schauenburg, Peter},
  year = {2021},
  month = jun,
}

@article{Bonderson_2018_beyond,
  journal = "",
  title = {On invariants of Modular categories beyond modular data},
  url = {http://arxiv.org/abs/1805.05736},
  DOI = {10.48550/arXiv.1805.05736},
  note = {arXiv:1805.05736 [math]},
  number = {arXiv:1805.05736},
  publisher = {arXiv},
  journal = "",
  author = {Bonderson, Parsa and Delaney, Colleen and Galindo, César and Rowell,
            Eric C. and Tran, Alan and Wang, Zhenghan},
  year = {2018},
  month = jun,
}

@article{sheffer2025extracting,
  title={Extracting topological spins from bulk multipartite entanglement},
  author={Sheffer, Yarden and Stern, Ady and Berg, Erez},
  journal={Physical Review Letters},
  volume={135},
  number={8},
  pages={086601},
  year={2025},
  doi = {10.1103/h53w-64dy},
  url = {https://link.aps.org/doi/10.1103/h53w-64dy},
  publisher={APS}
}

@article{kitaev2006anyons,
  title = {Anyons in an exactly solved model and beyond},
  author = {Kitaev, Alexei},
  journal = {Annals of Physics},
  volume = {321},
  number = {1},
  pages = {2--111},
  year = {2006},
  publisher = {Elsevier},
}

@article{siva_universal_2022,
  title = {A universal tripartite entanglement signature of ungappable edge
           states},
  volume = {106},
  issn = {2469-9950, 2469-9969},
  url = {http://arxiv.org/abs/2110.11965},
  doi = {10.1103/PhysRevB.106.L041107},
  language = {en},
  number = {4},
  urldate = {2024-08-14},
  journal = {Physical Review B},
  author = {Siva, Karthik and Zou, Yijian and Soejima, Tomohiro and Mong, Roger
            S. K. and Zaletel, Michael P.},
  month = jul,
  year = {2022},
  note = {arXiv:2110.11965 [quant-ph]},
  pages = {L041107},
}

@article{kim2022modular,
  title = {Modular commutator in gapped quantum many-body systems},
  volume = {106},
  ISSN = {2469-9950, 2469-9969},
  DOI = {10.1103/PhysRevB.106.075147},
  number = {7},
  journal = {Phys. Rev. B},
  author = {Kim, Isaac H. and Shi, Bowen and Kato, Kohtaro and Albert, Victor V.
            },
  year = {2022},
  month = aug,
  pages = {075147},
}

@article{gass2025r,
  title = {R$\backslash$'enyi-like entanglement probe of the chiral central
           charge},
  author = {Gass, Julian and Levin, Michael},
  journal = {arXiv preprint arXiv:2512.20608},
  year = {2025},
}

@article{del2026multipartite,
  title = {From Multipartite Entanglement to TQFT},
  author = {Del Zotto, Michele and Gadde, Abhijit and Putrov, Pavel},
  journal = {arXiv preprint arXiv:2602.16770},
  year = {2026},
}

@article{Burnell_Chen_Fidkowski_Vishwanath_2014,
  title = {Exactly Soluble Model of a 3D Symmetry Protected Topological Phase of
           Bosons with Surface Topological Order},
  volume = {90},
  ISSN = {1098-0121, 1550-235X},
  url = {http://arxiv.org/abs/1302.7072},
  DOI = {10.1103/PhysRevB.90.245122},
  note = {arXiv:1302.7072 [cond-mat]},
  number = {24},
  journal = {Phys. Rev. B},
  author = {Burnell, F. J. and Chen, Xie and Fidkowski, Lukasz and Vishwanath,
            Ashvin},
  year = {2014},
  month = dec,
  language = {en},
}

@article{Haah_Fidkowski_Hastings_2023,
  title = {Nontrivial Quantum Cellular Automata in Higher Dimensions},
  volume = {398},
  ISSN = {1432-0916},
  DOI = {10.1007/s00220-022-04528-1},
  number = {1},
  journal = {Commun. Math. Phys.},
  author = {Haah, Jeongwan and Fidkowski, Lukasz and Hastings, Matthew B.},
  year = {2023},
  month = feb,
  pages = {469–540},
  language = {en},
}

@article{Walker_Wang_2012,
  title = {(3+1)-TQFTs and topological insulators},
  volume = {7},
  ISSN = {2095-0470},
  DOI = {10.1007/s11467-011-0194-z},
  number = {2},
  journal = {Front. Phys.},
  author = {Walker, Kevin and Wang, Zhenghan},
  year = {2012},
  month = apr,
  pages = {150–159},
  language = {en},
}

@article{Basak_2016,
  title = {An algorithmic approach to construct crystallizations of 3-manifolds
           from presentations of fundamental groups},
  volume = {126},
  ISSN = {0253-4142, 0973-7685},
  DOI = {10.1007/s12044-016-0302-7},
  number = {4},
  journal = {Proc Math Sci},
  author = {Basak, Biplab},
  year = {2016},
  month = nov,
  pages = {629–653},
  language = {en},
}

@article{Delaney_Tran_2018,
  title = {A systematic search of knot and link invariants beyond modular data},
  url = {http://arxiv.org/abs/1806.02843},
  DOI = {10.48550/arXiv.1806.02843},
  note = {arXiv:1806.02843 [math]},
  number = {arXiv:1806.02843},
  journal = "",
  publisher = {arXiv},
  author = {Delaney, Colleen and Tran, Alan},
  year = {2018},
  month = jun,
}

@book{Lins_1995,
  address = {Singapore},
  series = {Series on knots and everything},
  title = {Gems, Computers and attractors for 3-manifolds},
  ISBN = {978-981-02-1907-9},
  publisher = {World Scientific},
  author = {Lins, Sóstenes},
  year = {1995},
  collection = {Series on knots and everything},
  language = {en},
}

@article{aschenbrenner20123,
  title = {3-manifold groups},
  author = {Aschenbrenner, Matthias and Friedl, Stefan and Wilton, Henry},
  journal = {arXiv preprint arXiv:1205.0202},
  year = {2012},
}

@misc{regina,
  author = {Benjamin A. Burton and Ryan Budney and William Pettersson and others
            },
  title = {Regina: Software for low-dimensional topology},
  howpublished = {{\tt http://regina-normal.github.io/}},
  year = {1999--2025},
}

@article{Kapustin_2014,
  title = {Symmetry Protected Topological Phases, Anomalies, and Cobordisms:
           Beyond Group Cohomology},
  url = {http://arxiv.org/abs/1403.1467},
  DOI = {10.48550/arXiv.1403.1467},
  note = {arXiv:1403.1467 [cond-mat]},
  number = {arXiv:1403.1467},
  publisher = {arXiv},
  author = {Kapustin, Anton},
  year = {2014},
  month = apr,
  language = {en},
}

@book{Nakahara_1990,
  title = {Geometry, Topology and Physics},
  ISBN = {0-7503-0606-8},
  url = {http://stacks.iop.org/0750306068},
  DOI = {10.1887/0750306068},
  journal = {Text},
  publisher = {IOP Publishing Ltd},
  author = {Nakahara, Mikio},
  year = {1990},
}

@article{grover2011entanglement,
  title = {Entanglement entropy of gapped phases and topological order in three
           dimensions},
  author = {Grover, Tarun and Turner, Ari M and Vishwanath, Ashvin},
  journal = {Physical Review B—Condensed Matter and Materials Physics},
  volume = {84},
  number = {19},
  pages = {195120},
  year = {2011},
  publisher = {APS},
}

@article{sheffer2025probing,
  title = {Probing chiral topological states with permutation defects},
  author = {Sheffer, Yarden and Fan, Ruihua and Stern, Ady and Berg, Erez and
            Ryu, Shinsei},
  journal = {Phys. Rev. B},
  volume = {113},
  issue = {19},
  pages = {195123},
  numpages = {27},
  year = {2026},
  month = {May},
  publisher = {American Physical Society},
  doi = {10.1103/3bz4-fbxr},
  url = {https://link.aps.org/doi/10.1103/3bz4-fbxr},
}

@article{chen2021exactly,
  title = {Exactly solvable lattice Hamiltonians and gravitational anomalies},
  author = {Chen, Yu-An and Hsin, Po-Shen},
  journal = {arXiv preprint arXiv:2110.14644},
  year = {2021},
}

@misc{delzotto2026multi,
      title={Multipartite Entanglement, TQFT Partition Functions and PL Bordism}, 
      author={Michele Del Zotto and Abhijit Gadde and Pavel Putrov},
      year={2026},
      eprint={2611.07000},
      archivePrefix={arXiv},
      primaryClass={hep-th},
      url={https://arxiv.org/abs/2611.07000}, 
}

@article{stallings1959kneser,
  author  = {Stallings, John R.},
  title   = {Grushko's theorem {II}, {K}neser's conjecture},
  journal = {Notices of the American Mathematical Society},
  volume  = {6},
  number  = {5},
  pages   = {531--532},
  year    = {1959},
  note    = {Abstract 559-165}
}

@misc{sheffer2026gemscode,
  note = {\url{https://github.com/ysheffer/locally-achiral-gems}},
}

@article{levitt1978existence,
  author  = {Levitt, Norman and Rourke, Colin},
  title   = {The existence of combinatorial formulae for characteristic classes},
  journal = {Transactions of the American Mathematical Society},
  volume  = {239},
  pages   = {391--397},
  year    = {1978}
}

\end{document}